\documentclass{ieeeaccessV2}
\usepackage{xcolor}
\AtBeginDocument{%
  \colorlet{accessblue}{black}%
  \colorlet{greycolor}{black}%
  \colorlet{grey}{black}%
}
\usepackage{algpseudocode}
\usepackage{cite}
\usepackage{amsmath,amssymb,amsfonts, amsthm}
\usepackage{algorithm}
\usepackage{hyperref}
\usepackage{graphicx}
\usepackage{textcomp}
\def\BibTeX{{\rm B\kern-.05em{\sc i\kern-.025em b}\kern-.08em
    T\kern-.1667em\lower.7ex\hbox{E}\kern-.125emX}}
\begin{document}
\doi{}

\makeatletter
\def\headerlogo{}
\def\headerlogoall{}
\makeatother

\makeatletter

\def\ps@headings{%
  \def\@oddfoot{\raisebox{8pt}{%
    \hbox to 0pc{\hbox to \textwidth{\mbox{}\hfill{\footerpagefont\thepage}}}%
  }}%
  \def\@evenfoot{\raisebox{8pt}{%
    \hbox to 0pc{\hbox to \textwidth{{\footerpagefont\thepage}\hfill}}%
  }}%
}
\makeatother

\makeatletter
\def\@oddfoot{}
\def\@evenfoot{}
\makeatother

\makeatletter

\def\@oddfoot{%
  \raisebox{8pt}{%
    \hbox to \textwidth{\hfill{\footerpagefont\thepage}}%
  }%
}

\def\@evenfoot{%
  \raisebox{8pt}{%
    \hbox to \textwidth{{\footerpagefont\thepage}\hfill}%
  }%
}

\makeatother

\newtheorem{theorem}{Theorem}%
\newtheorem{lemma}{Lemma}%
\newtheorem{prop}{Proposition}%
\newtheorem{definition}{Definition}
\newtheorem{remark}{Remark}
\newtheorem{problem}{Problem}
\newtheorem{cor}{Corollary}
\newtheorem{assumption}{Assumption}

\title{A Trainable-Embedding Quantum Physics-Informed Framework for Multi-Species Reaction–Diffusion Systems
}
\author{\uppercase{Ban Q. Tran}\authorrefmark{1,4}\IEEEmembership{Student Member, IEEE}, Nahid Binandeh Dehaghani\authorrefmark{2}
\IEEEmembership{Member, IEEE}, A. Pedro Aguiar \authorrefmark{3} \IEEEmembership{Senior Member, IEEE}, Rafal Wisniewski\authorrefmark{2} \IEEEmembership{Fellow, IEEE}, Susan Mengel \IEEEmembership{Member, IEEE} \authorrefmark{1}}
\address[1]{Department of Computer Science, Texas Tech University, Lubbock, USA}
\address[2]{Department of Electronic Systems, Aalborg University, Aalborg, Denmark}
\address[3]{SYSTEC-ARISE, Faculty of Engineering, University of Porto, Porto, Portugal
\address[4]{Department of Computing Fundamentals, FPT University, Hanoi, Vietnam}}
\tfootnote{The authors acknowledge the support of the Danish e-Infrastructure Consortium (DeiC) and the National Quantum Algorithm Academy (NQAA) through the Postdoctoral Scholarship 
 under the project ``Quantum-Driven Solutions for Multi-Agent Systems and Advanced Computation''. This work was also partially supported by UID/00147- Research Center for Systems and Technologies (SYSTEC) - and the Associate Laboratory Advanced Production and Intelligent Systems (ARISE, 10.54499/LA/P/0112/2020) funded by Fundação para a Ciência e a Tecnologia, I.P./ MCTES through the national funds. Furthermore, this research was supported, in part, by the Advanced Computing and Data Resource program, supported by the U.S. National Science Foundation (ACCESS), and Texas Tech University. The authors also acknowledge the research support provided by FPT University in Vietnam.}

\markboth
{Tran-Dehaghani \headeretal: Extended Trainable-Embedding Quantum PINNs for Multi-Species Reaction–Diffusion Dynamics}
{Tran-Dehaghani \headeretal: Extended Trainable-Embedding Quantum PINNs for Multi-Species Reaction–Diffusion Dynamics}

\corresp{Corresponding authors: Ban Q. Tran (email: bantran@ttu.edu-bantq3@fe.edu.vn), Nahid Binandeh Dehaghani (email: nahidbd@es.aau.dk)}

\begin{abstract}
Physics-informed neural networks (PINNs) and hybrid quantum-classical extensions provide a promising framework for solving partial differential equations (PDEs) by embedding physical laws directly into the learning process. In this work, we study embedding strategies for trainable embedding quantum physics-informed neural networks (TE-QPINNs) in the context of nonlinear reaction-diffusion (RD) systems. We introduce an extended TE-QPINN (x-TE-QPINN) architecture that supports both classical and fully quantum embeddings, enabling a controlled comparison between feedforward neural network-based feature maps and parameterized quantum circuit embeddings. The first architecture is the classical embedding feed-forward neural network-based TE-QPINN (FNN-TE-QPINN), while the latter variant is a purely quantum one, referred to as quantum embedding neural network-based TE-QPINN (QNN-TE-QPINN). The proposed framework employs hardware-efficient variational quantum circuits and species-specific readout operators to approximate coupled multi-field dynamics while enforcing governing equations, boundary conditions, and initial conditions through a physics-informed loss function. By isolating the embedding mechanism while keeping the variational ansatz, loss formulation, and optimization procedure fixed, we analyze the impact of embedding design on gradient structure, parameter scaling, and quantum resource requirements. Numerical experiments on one- and two-dimensional RD equations demonstrate that quantum embeddings can replace classical embeddings without degradation in solution accuracy and, in certain regimes, exhibit improved optimization behavior compared to classical PINNs and hybrid quantum models with fixed embeddings. These results provide architectural insight into hybrid quantum PDE solvers and inform the design of resource-efficient quantum physics-informed learning methods for near-term quantum engineering applications.
\end{abstract} 


\titlepgskip=-15pt

\maketitle

\section{Introduction}
\label{sec:introduction}

R{eaction}-diffusion systems constitute a fundamental class of PDEs for modeling spatiotemporal phenomena in biomedical science, including drug transport, intracellular signaling, morphogenesis, and tumor growth~\cite{murray2007mathematical,siepmann2012modeling,kondo2010reaction}. By coupling nonlinear reaction kinetics with spatial diffusion, RD models capture the interplay between local biochemical interactions and global transport effects, enabling predictive simulation of complex, multiscale biological dynamics. Owing to their ability to represent both spatial heterogeneity and temporal evolution, RD systems are widely used in analysis, simulation, and control-oriented modeling.

Despite their modeling power, the numerical solution of coupled nonlinear RD systems remains computationally challenging. Strong nonlinearities, stiffness, and multi-species coupling typically require fine spatiotemporal discretizations and iterative solvers, leading to high computational costs and limited scalability. While classical numerical methods can achieve high accuracy, their computational burden becomes prohibitive for large-scale simulations, parametric studies, or real-time applications. PINNs~\cite{raissi2019physics} have been proposed as a data-efficient alternative by embedding governing equations, boundary conditions, and initial conditions directly into the learning objective. However, classical PINNs often exhibit slow convergence and optimization difficulties when applied to nonlinear and coupled PDEs, particularly in stiff or multi-field settings~\cite{cuomo2022scientific,toscano2025pinns}.

Recent advances in quantum computing have motivated the exploration of hybrid quantum-classical algorithms for scientific computing. In particular, variational quantum algorithms based on parameterized quantum circuits (PQCs) have shown promise as expressive function approximators within near-term quantum hardware constraints. Building on this idea, quantum physics-informed nueral networks (QPINNs)~\cite{kyriienko2021solving} and Trainable Embedding QPINNs (TE-QPINNs)~\cite{berger2025trainable} integrate PQCs into physics-informed learning frameworks, enabling the approximation of PDE solutions in high-dimensional Hilbert spaces. A central architectural component in these approaches is the embedding mechanism that maps classical spatiotemporal coordinates into quantum states, which strongly influences expressivity, trainability, and quantum resource requirements.

In existing TE-QPINN formulations, the embedding stage is typically realized using a classical feedforward neural network (FNN) that generates trainable feature maps for angle encoding in a quantum circuit. While this hybrid design offers flexibility and reduces quantum resource demands, it raises a fundamental question from a quantum engineering perspective: to what extent can the embedding stage itself be made quantum-native without compromising performance or scalability? Addressing this question is essential for understanding architectural trade-offs in hybrid quantum PDE solvers and for guiding their deployment on near-term quantum platforms.

In this work, we investigate embedding strategies for TE-QPINNs in the context of coupled nonlinear reaction–diffusion systems. We consider identical discretization, loss formulation, and optimization settings across classical PINNs, FNN-TE-QPINNs with classical feedforward neural network (FNN) embeddings, and a fully quantum embedding variant referred to as QNN-TE-QPINN. In QNN-TE-QPINN, the classical embedding network is replaced by a PQC, yielding a fully quantum trainable feature map while preserving the variational ansatz and physics-informed loss.

\textbf{Contributions.} The main contributions of this work are summarized as follows:
\begin{itemize}
    \item We study embedding strategies for trainable embedding QPINNs applied to coupled nonlinear PDEs, with a focus on RD systems.
    \item We introduce \emph{QNN-TE-QPINN}, a TE-QPINN architecture employing a fully quantum neural network (QNN) as the embedding function, replacing the classical feedforward embedding used in prior work.
    \item We develop an end-to-end differentiable training pipeline that combines parameter-shift rules for quantum gradients with automatic differentiation for PDE residuals, ensuring physics consistency and scalability.
    \item We provide a controlled comparison between classical PINNs, FNN-TE-QPINNs with classical embeddings, and QNN-TE-QPINNs, analyzing convergence behavior, accuracy, parameter scaling, and quantum resource requirements.
\end{itemize}

Overall, this work contributes to the quantum engineering of hybrid PDE solvers by clarifying the role of embedding design in quantum physics-informed learning. Rather than aiming to demonstrate quantum advantage, the focus is on architectural insight and resource-aware design choices relevant for near-term quantum computing platforms.

\Figure[t](topskip=0pt, botskip=0pt, midskip=0pt)[width=2.0\columnwidth]{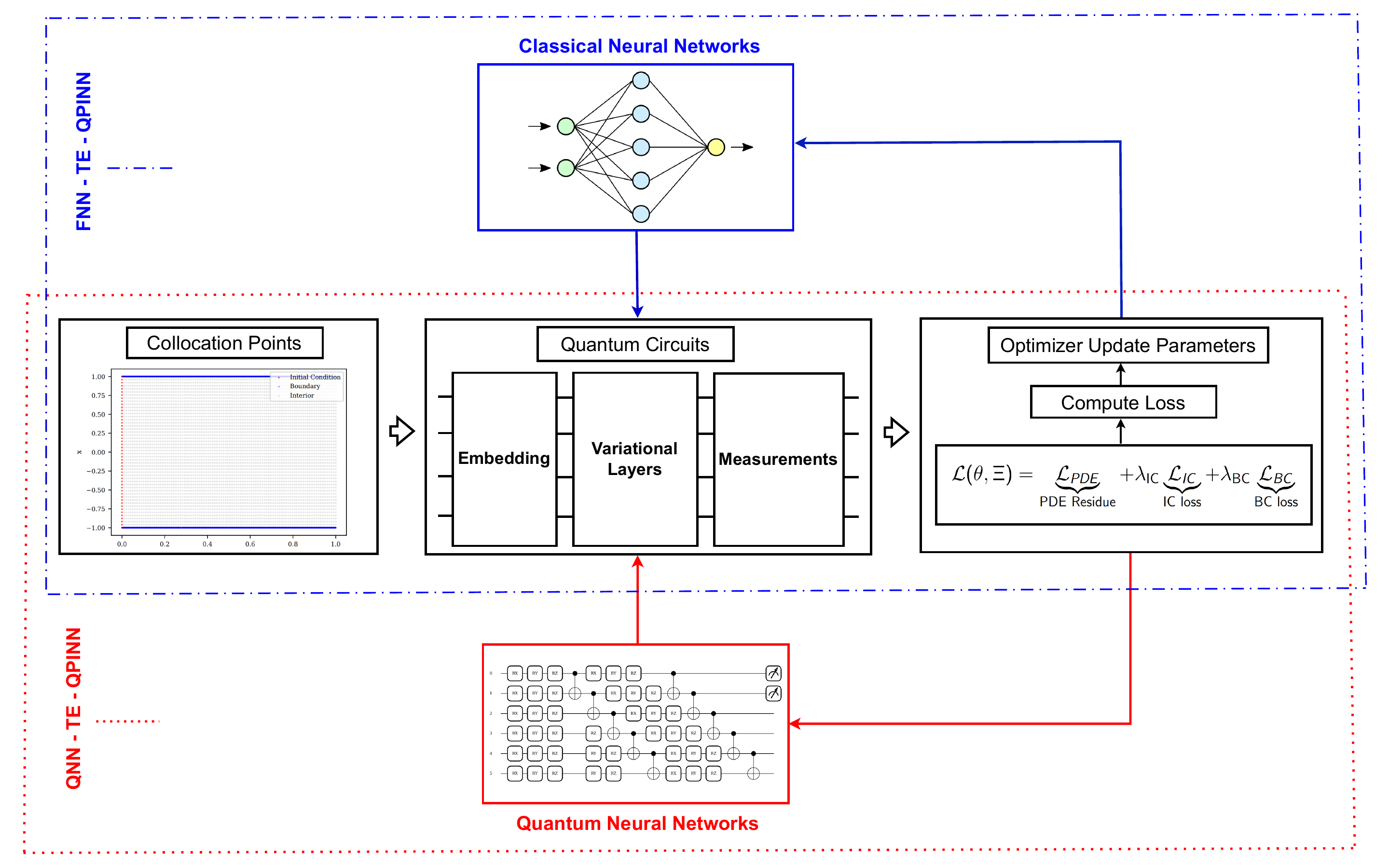}
{The overall architectural diagrams of the x-TE-QPINN, including the classical FNN-TE-QPINN and the purely quantum version QNN-TE-QPINN, where the embedding circuit is entirely constructed using a quantum circuit.
\label{xTE-QPINN-Arch}}

\section{Reaction-Diffusion Model Formulation}
\label{sec:rd-model}

The spatiotemporal evolution of interacting molecular species in an RD system is governed by coupled PDEs that combine passive diffusion with nonlinear reaction kinetics. As a representative benchmark, we consider a two-species RD model consisting of an autocatalytic activator \(A\) and a substrate \(S\), described by
\begin{align}
\frac{\partial c_A}{\partial t} &= D_A \nabla^2 c_A + \kappa_1 c_A^2 c_S - \kappa_2 c_A, \label{eq:dynamics_A}\\
\frac{\partial c_S}{\partial t} &= D_S \nabla^2 c_S - \kappa_1 c_A^2 c_S + \kappa_3, \label{eq:dynamics_S}
\end{align}
where \(c_A(x,t)\) and \(c_S(x,t)\) denote the concentrations of species \(A\) and \(S\) at spatial coordinate \(x\) and time \(t\), respectively. The parameters \(D_A\) and \(D_S\) represent diffusion coefficients, while \(\kappa_1\), \(\kappa_2\), and \(\kappa_3\) characterize autocatalytic amplification, linear decay of the activator, and external substrate replenishment~\cite{elf2003mesoscopic,murray2007mathematical}.

Equation~\eqref{eq:dynamics_A} governs the activator dynamics, which are characterized by slow diffusion and nonlinear amplification through interaction with the substrate. Equation~\eqref{eq:dynamics_S} describes the substrate dynamics, featuring faster diffusion, consumption through autocatalysis, and constant replenishment. The resulting diffusivity asymmetry (\(D_S \gg D_A\)) is a well-known mechanism underlying diffusion-driven instabilities and the emergence of spatially heterogeneous solutions, such as Turing patterns~\cite{turing1990chemical,cross1993pattern}.

The coupled system~\eqref{eq:dynamics_A}--\eqref{eq:dynamics_S} captures essential structural features common to nonlinear RD models: multiscale transport, nonlinear coupling between species, and competing reaction and diffusion timescales. The qualitative behavior of the system—including convergence to homogeneous steady states, temporal oscillations, or self-organized spatial structures—depends on the interplay between reaction rates and diffusion coefficients, often characterized by nondimensional quantities such as the Damköhler number~\cite{epstein1996nonlinear,kapral2012chemical}.

From a computational perspective, these properties make the model a challenging benchmark for physics-informed learning. Strong nonlinearities, coupled fields, and stiffness place significant demands on optimization and expressivity, particularly in data-scarce regimes. For these reasons, the RD system in~\eqref{eq:dynamics_A}--\eqref{eq:dynamics_S} provides a suitable testbed for evaluating hybrid quantum-classical physics-informed architectures.

This classical RD formulation establishes the physical foundation for the quantum-assisted learning framework introduced in the following sections, where the coupled dynamics are embedded into a physics-informed hybrid quantum-classical model for data-efficient approximation of the solution fields \((c_A,c_S)\).

\section{Quantum-Assisted Learning Formulation}
\label{sec:quantum_formulation}
We now introduce a hybrid quantum--classical physics-informed learning framework for approximating the solutions of coupled nonlinear RD systems. The proposed approach integrates collocation-based enforcement of governing equations, boundary conditions, and initial conditions with trainable variational quantum circuits, enabling a unified treatment of classical and quantum components within a single optimization pipeline. An overview of the resulting x-TE-QPINN architecture is shown in Fig.~\ref{xTE-QPINN-Arch}, which illustrates the interaction between classical preprocessing and residual evaluation, trainable embedding mechanisms, variational quantum circuit execution, species-specific readout, and parameter updates during training.

\subsection{Physics-Informed Problem Formulation}
Consider a coupled nonlinear RD system defined over a spatial domain $\Omega \subset \mathbb{R}^d$ and a temporal interval $t \in [0,T]$. The governing equations for the activator--substrate dynamics can be written in residual form as
\begin{equation*}
\begin{aligned}
\mathcal{D}_A[c_A, c_S; \boldsymbol{\beta}](x,t) - q_A(x,t) &= 0, 
&& (x,t) \in \Omega \times (0,T], \\
\mathcal{D}_S[c_A, c_S; \boldsymbol{\beta}](x,t) - q_S(x,t) &= 0, 
&& (x,t) \in \Omega \times (0,T], \\
\mathcal{B}[c_A, c_S](x,t) - \mathbf{b}(x,t) &= 0, 
&& (x,t) \in \partial\Omega \times (0,T], \\
\mathbf{c}(x,0) - \mathbf{c}_0(x) &= 0, 
&& x \in \Omega,
\end{aligned}
\end{equation*}
where $\mathbf{c}(x,t) = [\,c_A(x,t),\, c_S(x,t)\,]^{\mathsf{T}}$ denotes the
concentration vector of the two species and
$\boldsymbol{\beta} = [\,D_A, D_S, \kappa_1, \kappa_2, \kappa_3\,]^{\mathsf{T}}$
collects the physical parameters. The nonlinear operators
$\mathcal{D}_A[\cdot]$ and $\mathcal{D}_S[\cdot]$ encode the diffusive and
reactive dynamics defined in
Eqs.~\eqref{eq:dynamics_A}--\eqref{eq:dynamics_S}, while $q_i(x,t)$, $i\in\{A,S\}$, 
$\mathbf{b}(x,t)$, and $\mathbf{c}_0(x)$ specify source terms, boundary
conditions, and initial conditions, respectively.

\smallskip
The objective is to construct approximations
$\tilde{c}_A(x,t;\Theta)$ and $\tilde{c}_S(x,t;\Theta)$ that 
approximate 
the solution of 
the coupled RD system throughout the spatiotemporal domain. To this end, we introduce a hybrid quantum-classical physics-informed learning model parameterized by the composite set
$\Theta = (\theta_{\mathrm{var}}, \theta_{\mathrm{emb}})$,
where $\theta_{\mathrm{var}}$ denotes the variational parameters of the quantum circuit and $\theta_{\mathrm{emb}}$ parameterizes the embedding function. Depending on the embedding approachs, the embedding parameters are defined as
\[
\theta_{\mathrm{emb}} =
\begin{cases}
\theta_F, & \text{classical embedding (FNN-based)}, \\
\theta_Q, & \text{quantum embedding (QNN-based)}.
\end{cases}
\]

\smallskip
A trainable embedding function maps the spatiotemporal coordinates $(x,t)$ to a latent representation $\Gamma(x,t)$ suitable for quantum processing. In the classical embedding configuration, this mapping is performed by a feedforward neural network $\Gamma_{\theta_F}(x,t)$. In the quantum embedding configuration, the embedding is produced by a PQC with parameters $\theta_Q$. In both cases, the embedded features determine a data encoding operator $U_{\mathrm{enc}}(x,t;\theta_{\mathrm{emb}})$ by specifying input-dependent rotation angles applied prior to the variational evolution.

\smallskip
We define the composite unitary
$U(x,t)
\;\equiv\;
U_{\mathrm{var}}(\theta_{\mathrm{var}})\,
U_{\mathrm{enc}}(x,t;\theta_{\mathrm{emb}})$,
where $U_{\mathrm{var}}(\theta_{\mathrm{var}})$ denotes a data-independent variational quantum circuit and $U_{\mathrm{enc}}(x,t;\theta_{\mathrm{emb}})$ is a data-encoding operator whose rotation angles depend on the spatiotemporal coordinates $(x,t)$.
Predictions for the activator and substrate concentration fields are obtained from expectation values of species-specific observables. For
$i \in \{A,S\}$, the model output is given by
\begin{equation}
\tilde{c}_i(x,t)
=
\langle 0|^{\otimes N_q}
U^{\dagger}(x,t)\, 
\mathcal{O}_i\,
U(x,t) 
|0\rangle^{\otimes N_q},
\end{equation}
where $\mathcal{O}_i$ are Hermitian observables associated with the corresponding concentration fields and $N_q$ denotes the number of qubits. The spatiotemporal coordinates $(x,t)$ enter the quantum model exclusively through the data-encoding
operator $U_{\mathrm{enc}}(x,t;\theta_{\mathrm{emb}})$ as input-dependent rotation angles, while all trainable parameters are contained in $\Theta$.

\smallskip
The proposed architecture constitutes a physics-guided function approximator that learns the mapping
$(x,t) \longmapsto [\,\tilde{c}_A(x,t),\; \tilde{c}_S(x,t)\,]$
directly from the governing physical laws, without requiring labeled solution data. The coupled nature of the RD system is enforced through a unified physics-informed loss function that penalizes the PDE residuals, boundary conditions, and initial conditions for both species.

\smallskip
This formulation builds on recent advances in quantum-assisted scientific machine learning, where differentiable quantum circuits combined with trainable embeddings enhance expressivity and optimization stability
\cite{kyriienko2021solving,cerezo2021variational,bharti2022noisy}. Gradients of the hybrid model are computed using differentiable quantum nodes, combining parameter-shift rules for the quantum components with automatic differentiation
for the classical embedding networks \cite{dehaghani2025quantum}. This hybrid gradient strategy enables scalable training while maintaining consistency with the underlying physical constraints.

\subsection{Trainable Embedding Network}
\label{subsec:trainable_embedding}

To encode the spatiotemporal coordinates $(x,t)$ into quantum states suitable for variational processing, we employ a trainable embedding network that maps classical inputs to a latent representation controlling the quantum encoding
stage. This embedding determines the rotation angles applied to individual qubits prior to the variational evolution and constitutes the sole mechanism through which the input coordinates are incorporated into the quantum model.

\smallskip
Prior to embedding, the spatial and temporal coordinates are normalized to a compact interval,
\[
\tilde{x} = \mathcal{N}(x; x_{\min}, x_{\max}), \qquad
\tilde{t} = \mathcal{N}(t; t_{\min}, t_{\max}),
\]
where $\mathcal{N}(\cdot)$ denotes an affine normalization mapping
from the original coordinate range $[x_{\min}, x_{\max}]$
(and $[t_{\min}, t_{\max}]$, respectively) to the interval $[-1,1]$.
This normalization mitigates the scale imbalance between spatial and temporal dimensions, improving numerical stability during training.
The embedding function produces a feature map
\[
\Gamma(x,t)
=
\big[\,\alpha_1(\tilde{x},\tilde{t}),\;\ldots,\;
\alpha_{N_q}(\tilde{x},\tilde{t})\,\big]^{T},
\]
where each component $\alpha_i(\tilde{x},\tilde{t})$ specifies the rotation angle applied to the $i$-th qubit. The mapping $\Gamma(x,t)$ can be realized using different embedding strategies within a unified framework:
\begin{itemize}
    \item Quantum embedding (QNN-based): A PQC with variational parameters $\theta_Q$ processes the normalized coordinates $(\tilde{x},\tilde{t})$ and produces
    expectation values, which are subsequently transformed into
    rotation angles $\alpha_i(\tilde{x},\tilde{t};\theta_Q)$
    used in the data-encoding stage.
    \item Classical embedding (FNN-based): A lightweight feedforward neural network parameterized by $\theta_F$ 
    maps the normalized inputs
    $(\tilde{x},\tilde{t})$ directly to rotation angles
$\alpha_i(\tilde{x},\tilde{t};\theta_F)$, yielding a fully trainable classical feature map.
\end{itemize}
This modular design enables a controlled comparison between classical and fully quantum embeddings while preserving a common downstream variational architecture.

\smallskip
The embedding angles define the quantum encoding operator
$U_{\mathrm{enc}}(x,t;\theta_{\mathrm{emb}})
=
\bigotimes_{i=1}^{N_q}
R_y\!\big(\alpha_i(\tilde{x},\tilde{t}, \theta_{\mathrm{emb}})\big)$,
where $R_y(\cdot)=\exp(-i\,(\cdot)\,\sigma_y/2)$ denotes a single-qubit
rotation about the $y$-axis acting on the $i$-th qubit.
This operator
prepares the embedded quantum state
$|\psi_{\mathrm{enc}}(x,t)\rangle
=
U_{\mathrm{enc}}(x,t;\theta_{\mathrm{emb}})\,
|0\rangle^{\otimes N_q}$.
Depending on the selected embedding strategy, the angles $\alpha_i$ are produced by a classical neural network, a quantum embedding circuit, or an analytic mapping.
The embedding parameters $\theta_{\mathrm{emb}}$—instantiated as either $\theta_F$ or $\theta_Q$—are optimized jointly with the variational parameters $\theta_{\mathrm{var}}$ of the shared quantum circuit. This joint optimization allows the embedding to adapt to the spatiotemporal structure of the RD dynamics, while the repeated application of a fixed variational circuit enables a consistent functional representation over the entire domain. Importantly, this expressivity is achieved using shallow, hardware-efficient quantum circuits compatible with near-term quantum devices.

\subsection{Variational Quantum Circuit Design}
\label{subsec:vqc_design}

After the data-encoding stage, the prepared quantum state is processed by a parameterized variational quantum circuit (VQC), which constitutes the trainable quantum component of the proposed hybrid architecture. The role of the VQC is to transform the encoded quantum state into a representation from which the coupled concentration fields can be inferred through measurement.

\smallskip
The variational circuit is parameterized by a set of trainable variables $\theta_{\mathrm{var}}$ and is shared across all spatiotemporal inputs. That is, the same variational unitary is repeatedly applied to quantum states prepared with different data-dependent encodings. This shared-parameter design reduces the number of trainable quantum degrees of freedom and enforces a common latent quantum representation for the coupled RD dynamics. Multiple output fields are obtained by measuring distinct observables on the resulting quantum state.

\smallskip
The variational unitary is constructed as a sequence of $L$ layers,
$U_{\mathrm{var}}(\theta_{\mathrm{var}})
=
U_L(\theta_L)\cdots U_2(\theta_2)\,U_1(\theta_1)$,
where each layer $U_\ell(\theta_\ell)$ is expressed as a product of parameterized and unparameterized unitary operations,
\begin{equation}
U_\ell(\theta_\ell)
=
\prod_{m}
\exp\!\big(-i\,\theta_{\ell,m}\,H_m\big)\,W_m.
\label{eq:var_layer_abstract}
\end{equation}
Here, $H_m$ are fixed Hermitian generators and $W_m$ denote unparameterized unitary operators.

\smallskip
Each layer therefore consists of parameterized single-qubit operations followed by fixed entangling gates. In the present implementation, the generators $H_m$ are chosen from the Pauli basis and correspond to rotations about the $x$-, $y$-, and $z$-axes of the Bloch sphere applied independently to each qubit. The unparameterized unitaries $W_m$ implement a fixed nearest-neighbor entangling pattern realized as a chain of CNOT gates. This layered construction allows the expressivity of the variational circuit to be systematically increased by adding layers, while maintaining a shallow circuit depth and a
gate structure compatible with near-term quantum hardware.

\smallskip
For a given input $(x,t)$, the variational quantum state is obtained by applying the shared variational circuit to the encoded state prepared by the data encoding operator,
\begin{equation}
|\psi_{\mathrm{var}}(x,t;\Theta)\rangle
=
U_{\mathrm{var}}(\theta_{\mathrm{var}})
U_{\mathrm{enc}}(x,t;\theta_{\mathrm{emb}})
|0\rangle^{\otimes N_q}.
\end{equation}
Predictions for the activator and substrate concentration fields are then computed as expectation values of species-specific observables.

\smallskip
The variational parameters $\theta_{\mathrm{var}}$ are optimized jointly with the embedding parameters $\theta_{\mathrm{emb}}$ using the physics-informed loss introduced in the next subsection. By repeatedly applying a shared variational circuit to differently encoded quantum states, the model captures nonlinear coupling, spatial gradients, and temporal stiffness characteristic of RD systems, while preserving a hardware-efficient circuit structure suitable for near-term quantum implementations.

\subsection{Quantum Readout and Observable Design}
\label{subsec:quantum_readout}

The hybrid quantum model produces estimates of the activator and substrate concentration fields by evaluating expectation values of appropriately chosen Hermitian observables measured on the variational quantum state. Distinct physical fields are separated at the readout stage, while sharing a common latent quantum state generated by the embedding and variational circuits.

\smallskip
For each species $i \in \{A,S\}$, we associate a readout operator $\mathcal{O}_i$ acting on the $N_q$-qubit Hilbert space. In this work, we employ Pauli-$Z$–based observables of the form
$\mathcal{O}_i
=
\sum_{j \in \mathcal{Q}_i}
\sigma_z^{(j)}$,
where $\sigma_z^{(j)}$ denotes the Pauli-$Z$ operator acting on the $j$-th qubit and $\mathcal{Q}_i \subseteq \{1,\ldots,N_q\}$ specifies the subset of qubits associated with species $i$. Depending on the chosen readout configuration,
$\mathcal{Q}_i$ may contain a single qubit or multiple qubits, enabling flexible allocation of quantum resources across output fields.

\smallskip
In this work, we employ Pauli-$Z$--based observables, which
are efficiently measurable on near-term quantum hardware and yield bounded expectation values. This property is particularly advantageous in physics-informed learning, as bounded readouts improve the numerical stability of the PDE residuals and their gradients, thereby mitigating instabilities during training.
The predicted concentration fields are obtained as
\[
\tilde{c}_i(x,t)
=
\big\langle
\psi_{\mathrm{var}}(x,t)
\big|
\mathcal{O}_i
\big|
\psi_{\mathrm{var}}(x,t)
\big\rangle,
\qquad i \in \{A,S\},
\]
where $|\psi_{\mathrm{var}}(x,t)\rangle$ denotes the variational quantum state defined in Sec.~\ref{subsec:vqc_design}.

Although Pauli-$Z$–based observables provide a robust and low-overhead readout strategy, the proposed framework readily supports alternative Hermitian observables, such as weighted Pauli strings or problem-specific operators, provided they remain efficiently measurable on near-term devices. The choice of observable therefore reflects a trade-off between expressiveness, numerical stability, and practical implementability in quantum-assisted solutions of nonlinear RD systems.

\subsection{Physics-Informed Loss Function}
\label{subsec:physics_loss}

Training of the hybrid quantum-classical model is performed by minimizing a physics-informed loss function that enforces the governing RD equations together with the associated boundary and initial conditions. Let $\mathcal{S}_{\mathrm{PDE}} \subset \Omega \times (0,T]$ denote a set of interior collocation points, and let
$\mathcal{S}_{\mathrm{BC},k} \subset \partial\Omega_k \times (0,T]$ denote the collocation points associated with the $k$-th boundary segment. The total loss is defined as
\begin{equation}
\mathcal{L}(\Psi)
=
\mathcal{L}_{\mathrm{PDE}}(\Psi)
+
\sum_{k} \lambda_k\,\mathcal{L}_{\mathrm{BC},k}(\Psi)
+
\lambda_{\mathrm{IC}}\,\mathcal{L}_{\mathrm{IC}}(\Psi),
\label{eq:hybrid_loss}
\end{equation}
where $\Psi$ collects all trainable parameters of the embedding network and the variational quantum circuit, and $\lambda_k, \lambda_{\mathrm{IC}} \ge 0$ are weighting coefficients that balance the contributions of the different loss terms.

\smallskip
The interior loss penalizes violations of the coupled RD equations at the interior collocation points,
\begin{align}
&\mathcal{L}_{\mathrm{PDE}}(\Psi)
=
\sum_{(x^j,t^j)\in\mathcal{S}_{\mathrm{PDE}}}
\Big|
\mathcal{D}_A[\tilde{c}_A,\tilde{c}_S;\boldsymbol{\beta}](x^j,t^j)
-
q_A(x^j,t^j)
\Big|^{2}
\nonumber\\
&+
\sum_{(x^j,t^j)\in\mathcal{S}_{\mathrm{PDE}}}
\Big|
\mathcal{D}_S[\tilde{c}_A,\tilde{c}_S;\boldsymbol{\beta}](x^j,t^j)
-
q_S(x^j,t^j)
\Big|^{2}.
\end{align}
All spatial and temporal derivatives appearing in the differential operators $\mathcal{D}_A$ and $\mathcal{D}_S$ are evaluated using automatic differentiation through the hybrid quantum--classical computational graph. This includes differentiation of expectation values with respect to the input coordinates $(x,t)$ as well as differentiation with respect to the trainable parameters.

For each boundary segment $\partial\Omega_k$, the boundary loss enforces the corresponding boundary operator $\mathcal{B}_k[\cdot]$,
\begin{equation*}
\mathcal{L}_{\mathrm{BC},k}(\Psi)
=
\sum_{(x^j,t^j)\in\mathcal{S}_{\mathrm{BC},k}}
\Big|
\mathcal{B}_k[\tilde{c}_A,\tilde{c}_S](x^j,t^j)
-
b_k(x^j,t^j)
\Big|^{2}.
\end{equation*}

The initial condition at $t=0$ is enforced via
\begin{equation}
\mathcal{L}_{\mathrm{IC}}(\Psi)
=
\sum_{x^j\in\Omega}
\Big\|
\tilde{\mathbf{c}}(x^j,0;\Psi)
-
\mathbf{c}_0(x^j)
\Big\|^{2},
\qquad
\tilde{\mathbf{c}} = [\,\tilde{c}_A,\tilde{c}_S\,]^{\mathsf{T}}.
\end{equation}

\smallskip
The weighting coefficients $\lambda_k$ and $\lambda_{\mathrm{IC}}$ regulate the relative contributions of the governing equations, boundary constraints, and initial condition. Proper balancing of these terms is particularly important for stiff and nonlinear RD systems, where disparities in the magnitudes of residuals can impede convergence. In practice, appropriate tuning of the loss weights improves numerical stability and facilitates the accurate recovery of multiscale dynamics and sharp spatial gradients characteristic of RD processes.

\subsection{Gradient Structure and Derivative Computation}
\label{subsec:gradients}

Training the proposed hybrid quantum-classical model requires evaluating derivatives of the predicted concentrations with respect to both the variational quantum parameters and the embedding parameters. These derivatives enter the physics-informed loss through the differential operators of the RD system, which involve spatial and temporal derivatives of
the learned fields.

\smallskip
The predicted concentrations $\tilde{c}_i(x,t)$ are defined as expectation values of observables measured on a variational quantum state. Their dependence on the spatiotemporal coordinates $(x,t)$ arises exclusively through the embedding stage, which maps $(x,t)$ to a set of rotation angles $\alpha_j(x,t;\theta_{emb})$ applied to the quantum circuit. Consequently, differentiation with respect to spatial or temporal coordinates propagates through the embedding before affecting the quantum circuit output.

Applying the chain rule, the derivative of the model output with respect to an input coordinate (e.g.\ $x$) decomposes into a classical and a quantum contribution. The classical component captures the sensitivity of the embedding angles to changes in the input coordinate, while the quantum component quantifies the response of the circuit expectation value to infinitesimal
variations in those angles. Formally, one obtains
\begin{equation}
\frac{\partial \tilde{c}_i(x,t)}{\partial x}
=
\sum_{j=1}^{N_q}
\frac{\partial \tilde{c}_i}{\partial \alpha_j}
\,
\frac{\partial \alpha_j(x,t;\chi)}{\partial x},
\qquad i \in \{A,S\},
\label{eq:input_derivative}
\end{equation}
where the partial derivatives $\partial \tilde{c}_i / \partial \alpha_j$ are evaluated using the parameter-shift rule applied to the corresponding quantum rotations, and the derivatives $\partial \alpha_j / \partial x$ are computed via backpropagation through the embedding network \cite{bergholm2018pennylane,schuld2019evaluating}. An analogous expression holds for differentiation with respect to the temporal coordinate $t$.

\smallskip
Gradients of the total loss with respect to any trainable parameter $\theta \in \{\theta_{var},\theta_{emb}\}$ follow from repeated application of the chain rule. Specifically,
\begin{align}
\frac{\partial \mathcal{L}}{\partial \theta}
&=
\sum_{(x^j,t^j)\in\mathcal{S}_{\mathrm{PDE}}}
\sum_{i\in\{A,S\}}
\frac{\partial \mathcal{L}_{\mathrm{PDE}}}{\partial \tilde{c}_i(x^j,t^j)}
\,
\frac{\partial \tilde{c}_i(x^j,t^j)}{\partial \theta}
\notag\\[1mm]
&\quad+
\sum_{k} \lambda_k
\frac{\partial \mathcal{L}_{\mathrm{BC},k}}{\partial \theta}
+
\lambda_{\mathrm{IC}}
\frac{\partial \mathcal{L}_{\mathrm{IC}}}{\partial \theta}.
\label{eq:loss_grad_general}
\end{align}

\smallskip
Since the RD operators are known analytically, the derivatives of the PDE residuals with respect to the predicted concentrations $\tilde{c}_A$ and $\tilde{c}_S$ admit closed-form expressions. The remaining terms reduce to derivatives of quantum expectation values with respect to either variational circuit parameters or embedding parameters. These
derivatives are computed using a hybrid differentiation strategy that combines parameter-shift rules for the quantum components with automatic differentiation for the classical embedding network.

\smallskip
This separation of responsibilities ensures that quantum evaluations are required only for the variational circuit, while all embedding-related gradients are computed classically. As a result, the overall training procedure remains efficient and scalable, even for dense collocation grids and stiff nonlinear RD dynamics. To clarify the interpretation of the quantum derivatives appearing in the above expressions, we make the following remark.

\begin{remark}
Expressions involving derivatives with respect to individual quantum gates should be interpreted as shorthand for derivatives with respect to the corresponding gate parameters. The present formulation makes this dependence explicit by expressing all
quantum derivatives in terms of the corresponding continuous rotation parameters. 
\end{remark}

\begin{prop}[Gradient of the PDE Loss w.r.t.\ Embedding Parameters]
\label{prop:pde_grad_chi}
Let $R_A$ and $R_S$ denote the PDE residuals
\[
R_A=\mathcal{D}_A[\tilde c_A,\tilde c_S]-q_A,
\qquad
R_S=\mathcal{D}_S[\tilde c_A,\tilde c_S]-q_S,
\]
evaluated at a finite interior collocation set
$\mathcal{S}_{\mathrm{PDE}}\subset\Omega\times(0,T]$,
assumed to be a finite set of collocation points.
The PDE loss is
\[
\mathcal{L}_{\mathrm{PDE}}
=
\sum_{(x^j,t^j)\in\mathcal{S}_{\mathrm{PDE}}}
\big(R_A(x^j,t^j)^2 + R_S(x^j,t^j)^2\big).
\]
Define the residual sensitivities
\[
G_A
=
R_A\,\frac{\partial R_A}{\partial \tilde c_A}
+
R_S\,\frac{\partial R_S}{\partial \tilde c_A},
\qquad
G_S
=
R_A\,\frac{\partial R_A}{\partial \tilde c_S}
+
R_S\,\frac{\partial R_S}{\partial \tilde c_S}.
\]
Then the gradient of $\mathcal{L}_{\mathrm{PDE}}$ with respect to the embedding parameters $\theta_{emb}$ is given by
\begin{equation}
\label{eq:prop_final_chi_grad}
\begin{split}
\frac{\partial \mathcal{L}_{\mathrm{PDE}}}{\partial \theta_{emb}}
&=
2\sum_{(x^j,t^j)\in\mathcal{S}_{\mathrm{PDE}}}
\sum_{i\in\{A,S\}}
G_i(x^j,t^j)
\\
&\quad\times
\sum_{m=1}^{N_q}
\frac{\partial \tilde c_i}{\partial \alpha_m}
\frac{\partial \alpha_m( \tilde x^j,\tilde t^j;\theta_{emb})}{\partial \theta_{emb}}.
\end{split}
\end{equation}

\end{prop}

\begin{proof}
Differentiating 
$\mathcal{L}_{\mathrm{PDE}}$ with respect to $\theta_{emb}$ yields
\[
\frac{\partial \mathcal{L}_{\mathrm{PDE}}}{\partial \theta_{emb}}
=
2\sum_{(x^j,t^j)\in\mathcal{S}_{\mathrm{PDE}}}
\left(
R_A\,\frac{\partial R_A}{\partial \theta_{emb}}
+
R_S\,\frac{\partial R_S}{\partial \theta_{emb}}
\right).
\]
Each residual depends on $\theta_{emb}$ only through the predicted concentrations $\tilde c_A$ and $\tilde c_S$. Applying the chain rule gives
\[
\frac{\partial R_k}{\partial \theta_{emb}}
=
\frac{\partial R_k}{\partial \tilde c_A}
\frac{\partial \tilde c_A}{\partial \theta_{emb}}
+
\frac{\partial R_k}{\partial \tilde c_S}
\frac{\partial \tilde c_S}{\partial \theta_{emb}},
\qquad k\in\{A,S\}.
\]
Here, the derivatives $\partial R_k / \partial \tilde c_i$ are understood in the sense of functional derivatives induced by the differential operators $\mathcal{D}_A$ and $\mathcal{D}_S$, and are evaluated via automatic differentiation.
Substituting these expressions and grouping terms yields
\[
\frac{\partial \mathcal{L}_{\mathrm{PDE}}}{\partial \theta_{emb}}
=
2\sum_{(x^j,t^j)\in\mathcal{S}_{\mathrm{PDE}}}
\Big[
G_A\,\frac{\partial \tilde c_A}{\partial \theta_{emb}}
+
G_S\,\frac{\partial \tilde c_S}{\partial \theta_{emb}}
\Big].
\]
The predicted concentrations depend on $\theta_{emb}$ only through the embedding
angles $\alpha_m(\tilde x,\tilde t;\theta_{\mathrm{emb}})$. Applying the chain rule once more gives
\[
\frac{\partial \tilde c_i}{\partial \theta_{emb}}
=
\sum_{m=1}^{N_q}
\frac{\partial \tilde c_i}{\partial \alpha_m}
\frac{\partial \alpha_m(\tilde x,\tilde t;\theta_{\mathrm{emb}})}{\partial \theta_{emb}},
\qquad i\in\{A,S\},
\]
where the partial derivatives
$\partial \tilde c_i / \partial \alpha_m$ are evaluated using the
parameter-shift rule. Substituting this expression completes the proof.
\end{proof}

\begin{prop}[Gradient of the PDE Loss w.r.t.\ Variational Parameters]
\label{prop:pde_grad_psi}
Let $\mathcal{L}_{\mathrm{PDE}}$ be the PDE loss defined in
\eqref{eq:hybrid_loss}, and let $G_A$ and $G_S$ denote the residual sensitivity
terms defined in Proposition~\ref{prop:pde_grad_chi}. Let
$\theta_{var} = (\theta_1,\ldots,\theta_{N_{var} })$ denote the variational parameters of the quantum circuit $U_{\mathrm{var}}$.
Then, for each variational parameter $\theta_\ell$, the gradient of the PDE loss is given by
\begin{equation}
\label{eq:prop_final_psi_grad}
\frac{\partial \mathcal{L}_{\mathrm{PDE}}}{\partial \theta_\ell}
=
2\sum_{(x^j,t^j)\in\mathcal{S}_{\mathrm{PDE}}}
\sum_{i\in\{A,S\}}
G_i(x^j,t^j)
\frac{\partial \tilde c_i(x^j,t^j)}{\partial \theta_\ell},
\end{equation}
for $\ell=1,\ldots,N_{\mathrm{var}}$ where each partial derivative
$\partial \tilde c_i / \partial \theta_\ell$
is evaluated using the parameter-shift rule applied to the variational circuit 
$U_{\mathrm{var}}$.
\end{prop}

\begin{proof}
The PDE loss $\mathcal{L}_{\mathrm{PDE}}$ is a sum of squared residuals and
depends on the variational parameters $\theta_{var}$ only through the predicted
concentrations $\tilde c_A$ and $\tilde c_S$. Differentiating each residual
with respect to a variational parameter $\theta_\ell$ yields
\[
\frac{\partial R_k}{\partial \theta_\ell}
=
\frac{\partial R_k}{\partial \tilde c_A}
\frac{\partial \tilde c_A}{\partial \theta_\ell}
+
\frac{\partial R_k}{\partial \tilde c_S}
\frac{\partial \tilde c_S}{\partial \theta_\ell},
\qquad k \in \{A,S\},
\]
by the chain rule.

Substituting these expressions into the derivative of the squared residuals and collecting terms yields weighted sums of the derivatives $\partial \tilde c_i / \partial \theta_\ell$, where the weights are precisely the residual sensitivity terms $G_A$ and $G_S$ defined in Proposition~\ref{prop:pde_grad_chi}. The resulting expression therefore depends linearly on the derivatives of the predicted concentrations with respect to $\theta_\ell$.

Since the variational parameters appear exclusively in the unitary $U_{\mathrm{var}}$, each derivative $\partial \tilde c_i / \partial \theta_\ell$ corresponds to differentiating a quantum expectation value with respect to a single circuit parameter and is evaluated using the parameter-shift rule. This establishes the stated result.
\end{proof}

\begin{remark}
Propositions~\ref{prop:pde_grad_chi} and~\ref{prop:pde_grad_psi} jointly demonstrate that the gradient of the physics-informed loss admits an additive decomposition into contributions arising from the embedding and from the variational quantum circuit, thereby enabling efficient hybrid quantum-classical optimization.
\end{remark}

\begin{remark}[Unified Gradient Structure]
\label{rem:grad_structure}
For any trainable parameter $\theta\in\{\theta_{\mathrm{emb}},\theta_{\mathrm{var}}\}$ of the hybrid model, the gradient of the PDE loss admits the unified representation
\[
\frac{\partial \mathcal L_{\mathrm{PDE}}}{\partial \theta}
=
2 \sum_{(x_j,t_j)\in\mathcal S_{\mathrm{PDE}}}
\sum_{i\in\{A,S\}}
G_i(x_j,t_j)\,
\frac{\partial \tilde c_i(x_j,t_j)}{\partial \theta}.
\]
This unified form enables a single gradient evaluation pipeline in which the residual sensitivities are shared, while the parameter-specific derivatives are supplied by either Proposition~\ref{prop:pde_grad_chi} or Proposition~\ref{prop:pde_grad_psi}.
\end{remark}

This decomposition highlights the separation between classical and quantum
computational responsibilities and supports a modular implementation of the
hybrid training procedure.

\begin{algorithm}[t]
\small
\caption{Training Algorithm for the x-TE-QPINN Framework}
\label{alg:xteqpinn} 

\begin{algorithmic}[1]
\State \textbf{Input:}
\Statex \quad Reaction-diffusion operators $\mathcal{D}_i[\cdot]$ and source terms $q_i$, $i\in\{A,S\}$
\Statex \quad Boundary operators $\mathcal{B}_k[\cdot]$ and boundary data $b_k$
\Statex \quad Initial condition $c_0(x)$

\Statex \quad Collocation sets $\mathcal{S}_{\mathrm{PDE}}, \{\mathcal{S}_{\mathrm{BC},k}\}, \mathcal{S}_{\mathrm{IC}}$

\Statex \quad Loss weights $\{\lambda_k\}, \lambda_{\mathrm{IC}}$

\Statex \quad Maximum iterations $N_{\max}$, tolerance $\varepsilon$

\Statex \quad Embedding type: classical (FNN-based) or quantum (QNN-based)

\State \textbf{Initialization:}
\State Initialize embedding parameters $\theta_{\mathrm{emb}}$
\State Initialize variational circuit parameters $\theta_{\mathrm{var}}$
\State Select readout observables $\mathcal{O}_A, \mathcal{O}_S$
\State Normalize the input coordinates $(x,t)$ to the scaled variables $(\tilde{x},\tilde{t}) \in [-1,1]^2$
\State Initialize optimizer for parameters $(\theta_{\mathrm{emb}}, \theta_{\mathrm{var}})$

\While{not converged}

\State \textbf{Forward evaluation}
\For{each collocation point $(x^j,t^j)$}
    \State Compute embedding angles $\alpha(x^j,t^j;\theta_{\mathrm{emb}})$    
    \State Prepare encoded state $|\psi_{\mathrm{enc}}\rangle = U_{\mathrm{enc}}(\alpha)|0\rangle^{\otimes N_q}$    
    \State Apply variational circuit $|\psi_{\mathrm{var}}\rangle = U_{\mathrm{var}}(\theta_{\mathrm{var}})\,|\psi_{\mathrm{enc}}\rangle$
  
\State Evaluate predicted fields
    $\tilde c_i(x^j,t^j)=\langle\psi_{\mathrm{var}}|\mathcal{O}_i|\psi_{\mathrm{var}}\rangle$
    \State Compute required spatial and temporal derivatives via the hybrid chain rule
\EndFor

\State \textbf{Loss evaluation}
\State Compute $\mathcal{L}_{\mathrm{PDE}}, \mathcal{L}_{\mathrm{BC}}, \mathcal{L}_{\mathrm{IC}}$
\State $\mathcal{L}\leftarrow \mathcal{L}_{\mathrm{PDE}}+\sum_k\lambda_k\mathcal{L}_{\mathrm{BC},k}+\lambda_{\mathrm{IC}}\mathcal{L}_{\mathrm{IC}}$

\State \textbf{Gradient computation}
\State Compute $\nabla_{\theta_{\mathrm{emb}}}\mathcal L$ (Proposition~1)

\State Compute $\nabla_{\theta_{\mathrm{var}}}\mathcal L$ (Proposition~2)

\State \textbf{Parameter update}
\State Update parameters $({\theta_{\mathrm{emb}}},{\theta_{\mathrm{var}}} )$ using the chosen optimizer

\EndWhile

\State \textbf{Output:}
\State Trained parameters $(\theta_{\mathrm{emb}}^\star,\theta_{\mathrm{var}}^\star)$ and learned solution fields $\tilde c_A(x,t), \tilde c_S(x,t)$

\end{algorithmic}
\end{algorithm}

\begin{remark}[On circuit architecture: single vs.\ species-specific circuits] In practice, two principal architectural choices arise for QPINNs applied to multi-species RD systems:
(i) a single variational quantum circuit with shared parameters and multi-output readout, independently of the chosen embedding strategy, as adopted in Algorithm~\ref{alg:xteqpinn}, and (ii) species-specific variational circuits acting on a shared encoded quantum state. The former is more resource-efficient, requiring only a single variational circuit per collocation point (up to parameter-shift evaluations), but induces coupling between species through shared parameters. The latter increases quantum resource requirements but allows each species to be modeled with independent parameters, improving expressivity and interpretability and more closely reflecting the structure of coupled RD systems. The choice between these architectures should be guided by the complexity of the dynamics, available quantum resources, and the desired balance between flexibility and computational efficiency.
\end{remark}

\Figure[t!](topskip=0pt, botskip=0pt, midskip=0pt)[width=2.0\columnwidth]{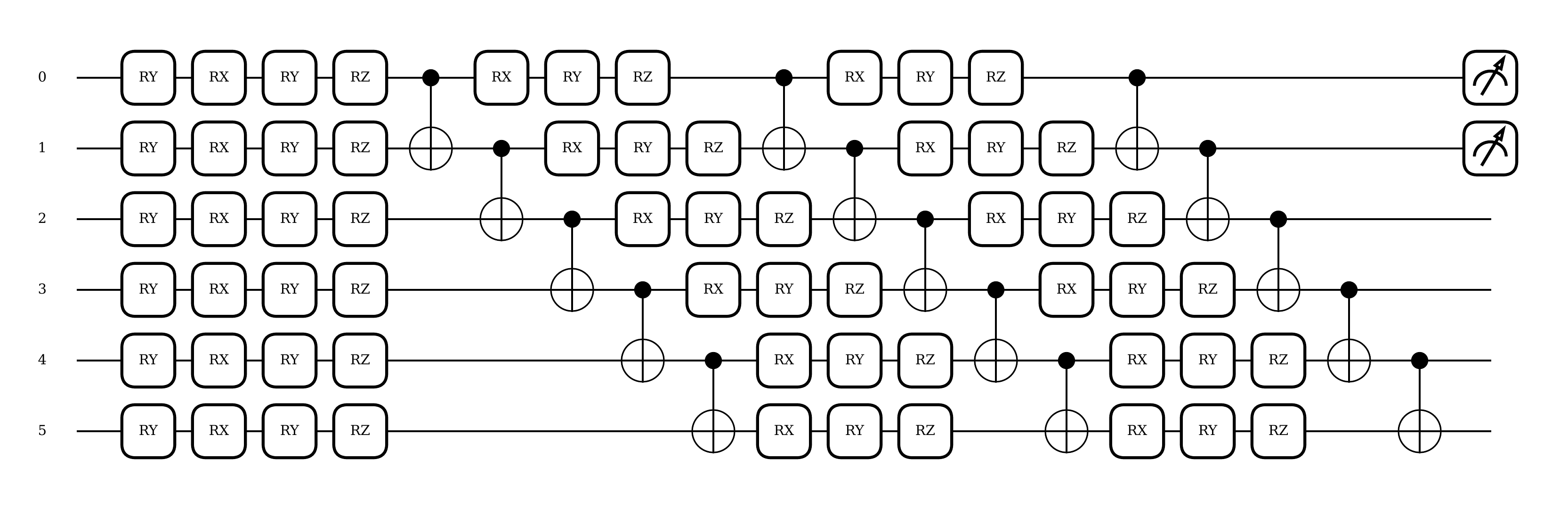}
{A quantum circuit architecture of the PQC layers used in our experiment, also known as the ansatz, serves as the heart of the x-TE-QPINN framework, analogous to the layers in classical neural networks, which consists of six qubits and three layers. This specific design utilizes a hardware-efficient ansatz, which is recognized as one of the most optimized architectures for learnable quantum circuits, particularly in maximizing performance within hardware constraints.
\label{pqc_layer}}

\Figure[t!](topskip=0pt, botskip=0pt, midskip=0pt)[width=1.8\columnwidth]{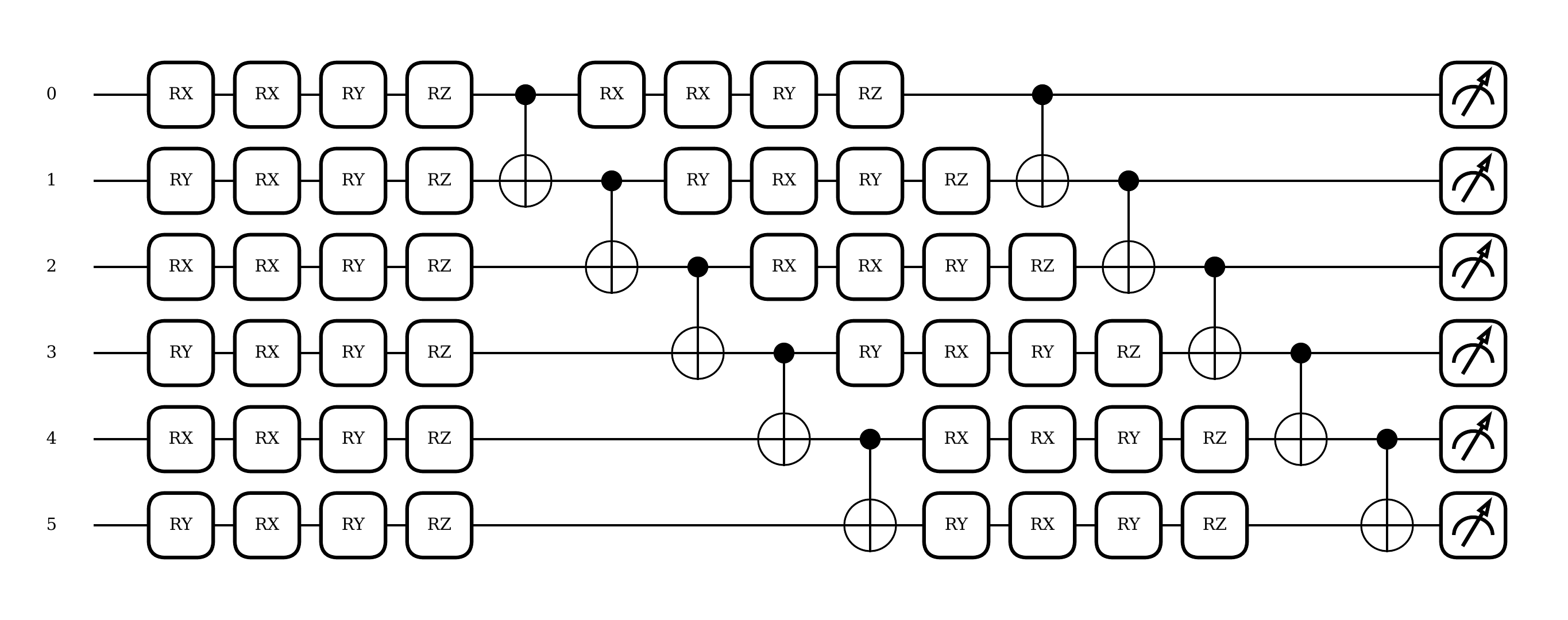}
{A quantum circuit sample for the QNN-based embedding function consists of six qubits and two layers for the quantum transformation stage. Throughout the experiments, the number of qubits and layers in the QNN were specifically tuned to match the parameter count of the PQC layers.
\label{embedding_layer}}


\subsection{Computational Cost and Parameter Scaling}

Training x-TE-QPINN is performed by minimizing a physics-informed loss over a set of $N_c$ collocation points using gradient-based optimization. The resulting computational cost per training iteration naturally decomposes into classical and quantum components. The dominant quantum cost arises from evaluating gradients of quantum expectation values using the parameter-shift rule. These evaluations are required for computing derivatives of the predicted concentration fields with respect to the variational circuit parameters and, in the case of a quantum embedding, with respect to the embedding parameters as well. As a result, the number of quantum circuit evaluations per training iteration scales linearly (up to constant factors arising from parameter-shift evaluations) with the number of trainable quantum parameters and with the number of collocation points.

In particular, for the hardware-efficient variational circuit considered here, acting on $n$ qubits with $L$ layers, the number of trainable variational parameters scales linearly in $nL$. When a quantum embedding is employed, additional quantum evaluations are required for the embedding circuit parameters. By contrast, when a classical feedforward neural network is used for the embedding, all embedding-related gradients are computed via standard backpropagation and do not incur additional quantum circuit evaluations.

\Figure[t!](topskip=0pt, botskip=0pt, midskip=0pt)[width=2.0\columnwidth]{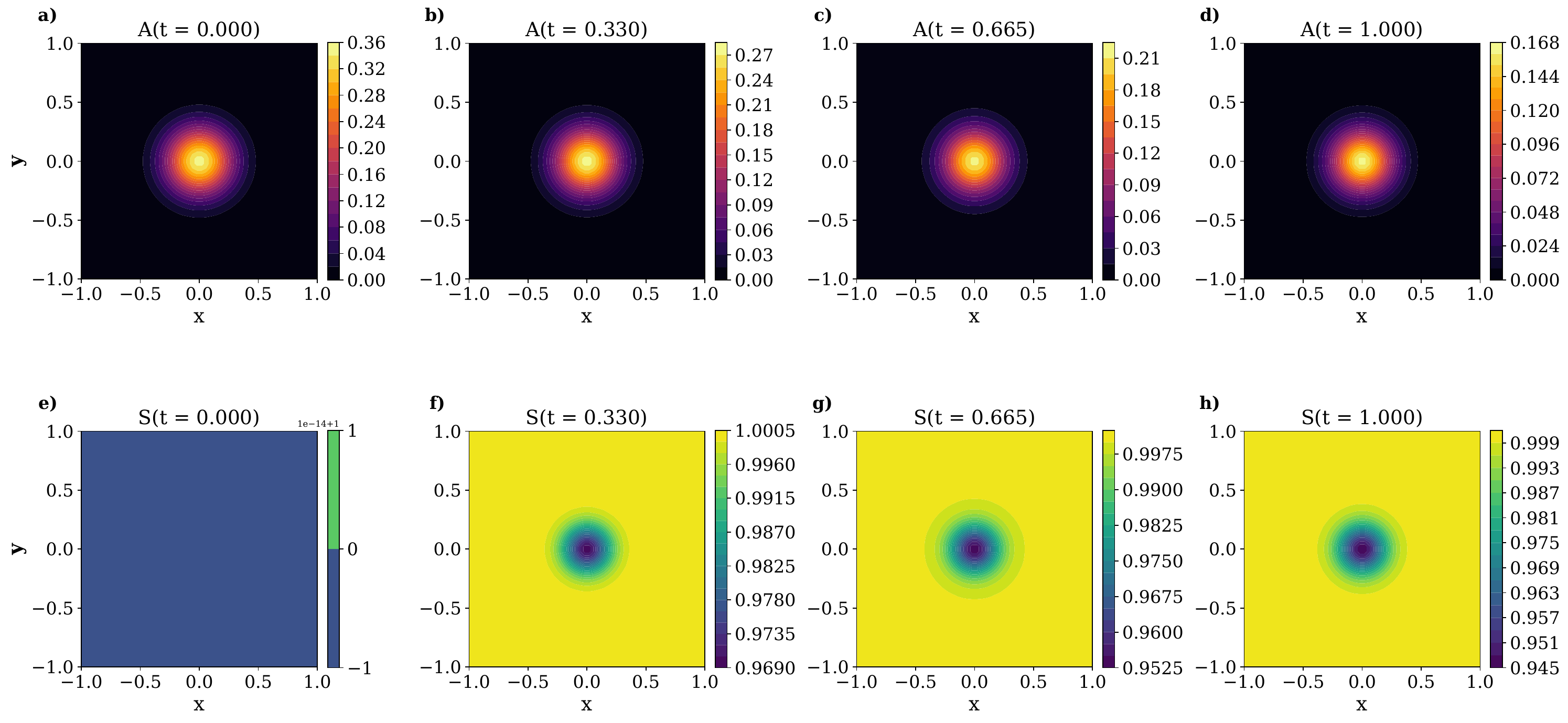}
{The exact reference solution of the two-dimensional RD equation was solved using the Runge-Kutta 45 solver with the Gaussian initial condition and periodic boundary conditions at time t=0, t = 0.330, t= 0.665 and t = 1.
\label{ref-sol-2d}}

\Figure[t!](topskip=0pt, botskip=0pt, midskip=0pt)[width=0.9\columnwidth]{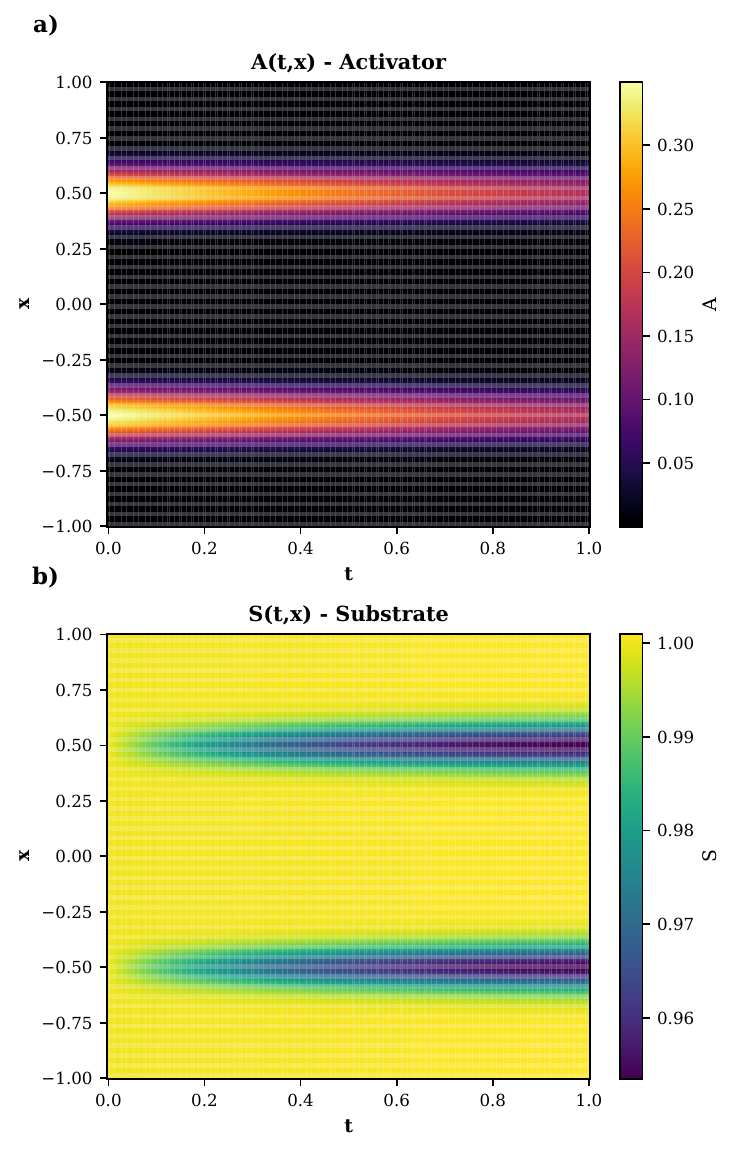}
{The exact reference solution of the one-dimensional RD equation was solved using the Runge-Kutta 45 solver with the double bump initial condition and periodic boundary conditions.
\label{ref-sol-1d}}

The classical computational cost is dominated by collocation sampling, evaluation of PDE residuals, boundary and initial-condition terms, and automatic differentiation through the embedding network. These operations scale polynomially with the number of collocation points and with the size of the classical embedding model. This clear separation of responsibilities enables a flexible hybrid training procedure in which quantum resources are reserved for representing and differentiating variational quantum states, while classical computation manages data handling, residual evaluation, and embedding optimization. In addition to per-iteration computational cost, the overall scalability of the x-TE-QPINN framework is determined by the growth of the model’s trainable parameter count.

The total number of trainable parameters in the x-TE-QPINN framework grows polynomially with the problem size and the chosen model architecture. The number of variational quantum parameters scales linearly with the number of qubits and the depth of the variational circuit, consistent with hardware-efficient ans\"atze designed for near-term quantum devices. When a classical embedding network is employed, the total parameter count includes both the variational circuit parameters and the weights of the embedding network, which scale with the width and depth of the chosen feedforward neural network.

\Figure[t!](topskip=0pt, botskip=0pt, midskip=0pt)[width=2.0\columnwidth]{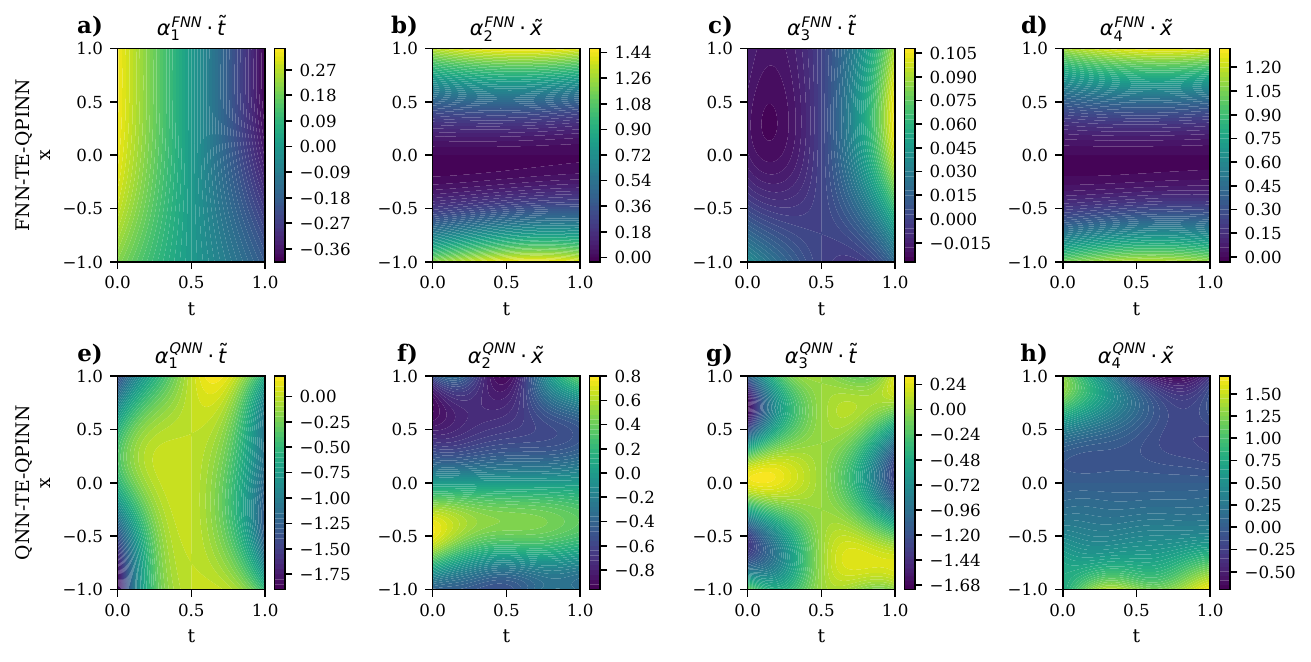}
{The output of activator A's embedding function of FNN and QNN for the one-dimensional RD equation, wherein the embedding function $\alpha$ is then further multiplied by $\tilde{x}$ or $\tilde{t}$ to enhance the quality of the training process for the subsequent layers of the x-TE-QPINN.
\label{emb}}

When a quantum embedding is used instead (QNN-TE-QPINN), the embedding parameters correspond to the gate parameters of a secondary, shallow and fixed-depth quantum circuit. Importantly, in both cases, the number of quantum parameters remains moderate and does not scale exponentially with the exponential dimension of the underlying Hilbert space. This polynomial parameter scaling ensures that the model remains trainable using standard gradient-based optimizers and allows the architecture to be adapted to available quantum and classical computational resources.

\subsection{Representational Capacity and Potential Quantum Advantages}
\label{subsec:quantum_representation}

Each variational quantum circuit in the x-TE-QPINN framework operates on an $n$-qubit Hilbert space $(\mathbb{C}^2)^{\otimes n}$ of dimension $2^n$. While this exponential dimension does not imply efficient exploration or coverage of the full state space, it provides a high-dimensional representational substrate, i.e., a $2^n$-dimensional Hilbert space of
quantum states that serves as the function representation space for the variational circuit,
in which nonlinear transformations can be implemented using a comparatively compact set of trainable parameters. This contrast between the exponential dimension of the representation space and the polynomial scaling of the parameter count highlights a key structural property of the x-TE-QPINN architecture: access to high-dimensional representations through a structured and resource-efficient parameterization.

\Figure[t!](topskip=0pt, botskip=0pt, midskip=0pt)[width=0.9\columnwidth]{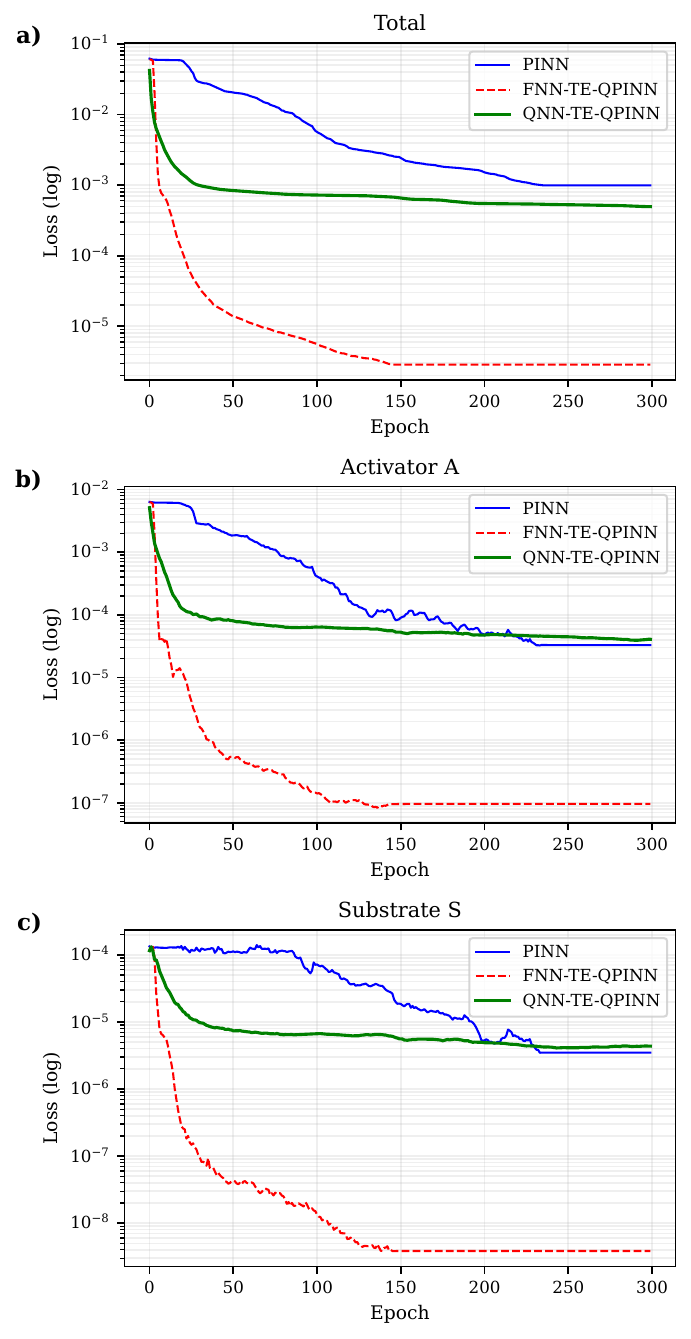}
{A comparison graph of the performance of three models: PINN, FNN-TE-QPINN, and QNN-TE-QPINN for the one-dimensional RD equation. The two QPINN models are configured with four qubits and ten layers of PQC. All three models are trained for 300 epochs on REPACSS: a) Total loss calculated for both species A (Activator) and S (Substrate). b) Loss calculated separately for the activator A. c) Loss calculated separately for the substrate S.
\label{train_perf_three}}

\Figure[t!](topskip=0pt, botskip=0pt, midskip=0pt)[width=0.9\columnwidth]{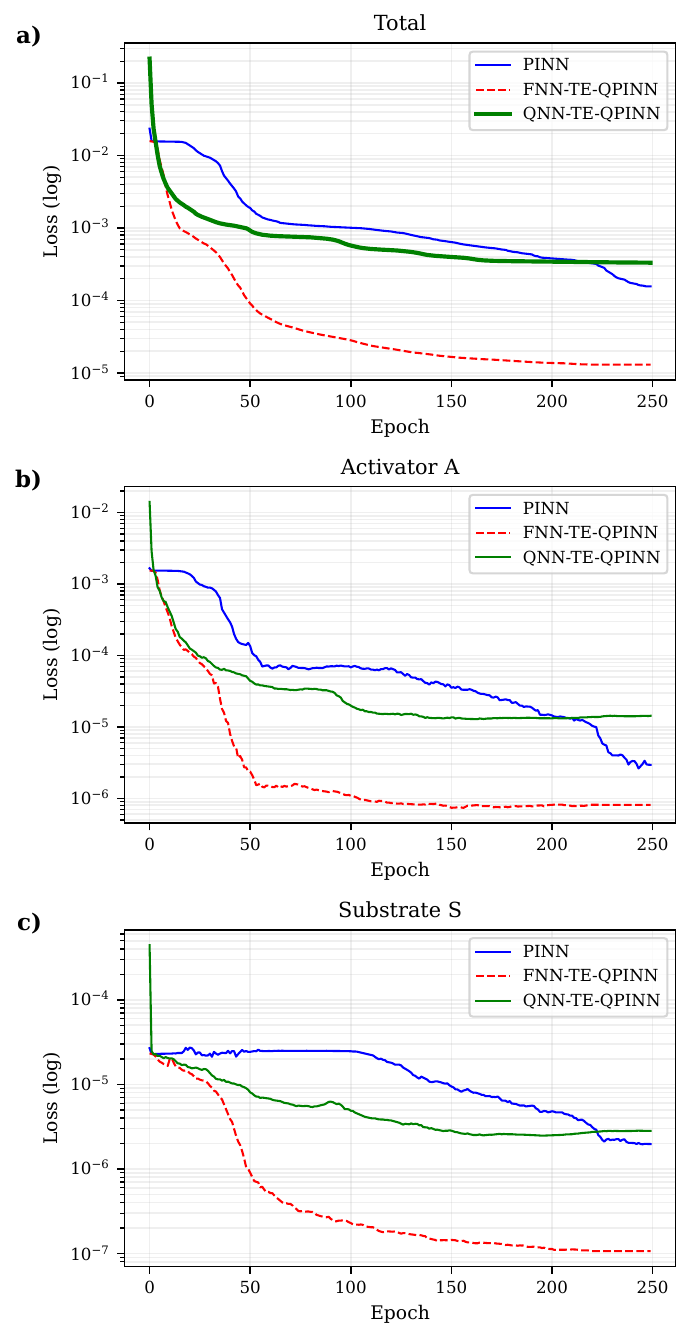}
{A comparison graph of the performance of three models—PINN, FNN-TE-QPINN, and QNN-TE-QPINN—after training for 250 epochs for the two-dimensional RD equation. Both the FNN-TE-QPINN and QNN-TE-QPINN models utilize four qubits and ten PQC layers for the main quantum circuit. Specifically, QNN-TE-QPINN also employs an additional four qubits and five PQC layers for its embedding component. a) Total loss. b) Loss of A. c) Loss of S.
\label{rd_2d_model_comparison}}

Quantum circuits naturally implement nonlinear transformations through parameterized unitary operations acting on such high-dimensional spaces, which may offer expressive representations for complex solution manifolds. While trainability and optimization remain nontrivial, this structural expressivity may be beneficial for learning solution landscapes that are challenging to approximate using classical neural networks of comparable size. Such properties are particularly relevant for RD systems exhibiting multiscale behavior, sharp spatial gradients, and strong nonlinear interactions between species.

At the same time, the hybrid architecture of x-TE-QPINN enables classical and quantum computational resources to be allocated adaptively. Classical computation handles collocation sampling, normalization, residual evaluation, and embedding-network training, while quantum circuits focus on representing and differentiating variational quantum states. As quantum hardware continues to improve in terms of qubit counts, coherence times, and gate fidelities, the structural properties identified in this work may become increasingly relevant for quantum-assisted scientific machine learning. The x-TE-QPINN framework thus provides a principled approach for integrating variational quantum models into physics-informed learning pipelines, offering a well-defined trade-off between computational cost, representational capacity, and model flexibility.

\section{Experimental Environment and Result Analysis}

\Figure[t!](topskip=0pt, botskip=0pt, midskip=0pt)[width=2.0\columnwidth]{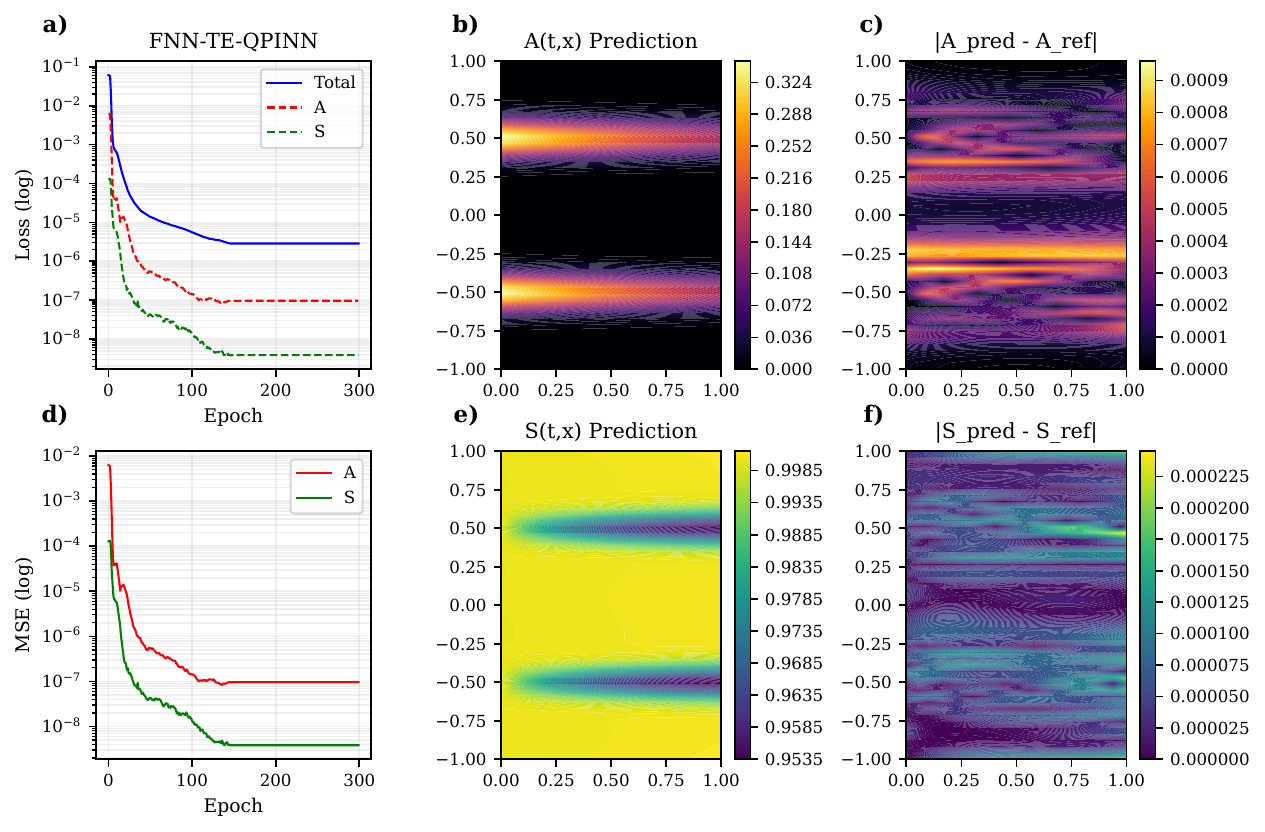}
{Results of solving the one-dimensional RD equation using FNN-TE-QPINN during the training session with 300 epochs: a) \& d) Total loss, individual loss, and MSE of FNN-TE-QPINN individual loss for activator A and substrate S species. b) \& e)  Predicted result for A and S by FNN-TE-QPINN. c) \& f) Absolute error between the FNN-TE-QPINN prediction for A and S with the reference solution (solved using RK45).
\label{TE-QPINN-sol}}

In this section, we present the experimental design process and the results obtained from the experiments conducted in this research. Fig. \ref{pqc_layer} illustrates the quantum circuit for the learning component, designed to function similarly to a classical neural network. This circuit incorporates eight qubits for the embedding stage and a PQC with a depth of 5 layers. Fig. \ref{embedding_layer} illustrates a specific quantum circuit designed for the embedding process, serving a function analogous to that of a classical feed-forward neural network (FNN) block. This particular circuit utilizes two input qubits and has a depth of 3 layers. Throughout the experiments, the number of qubits and layers in the QNN were specifically tuned to match the parameter count of the PQC layers.

\subsection{Experiment Setup}

\Figure[t!](topskip=0pt, botskip=0pt, midskip=0pt)[width=2.0\columnwidth]{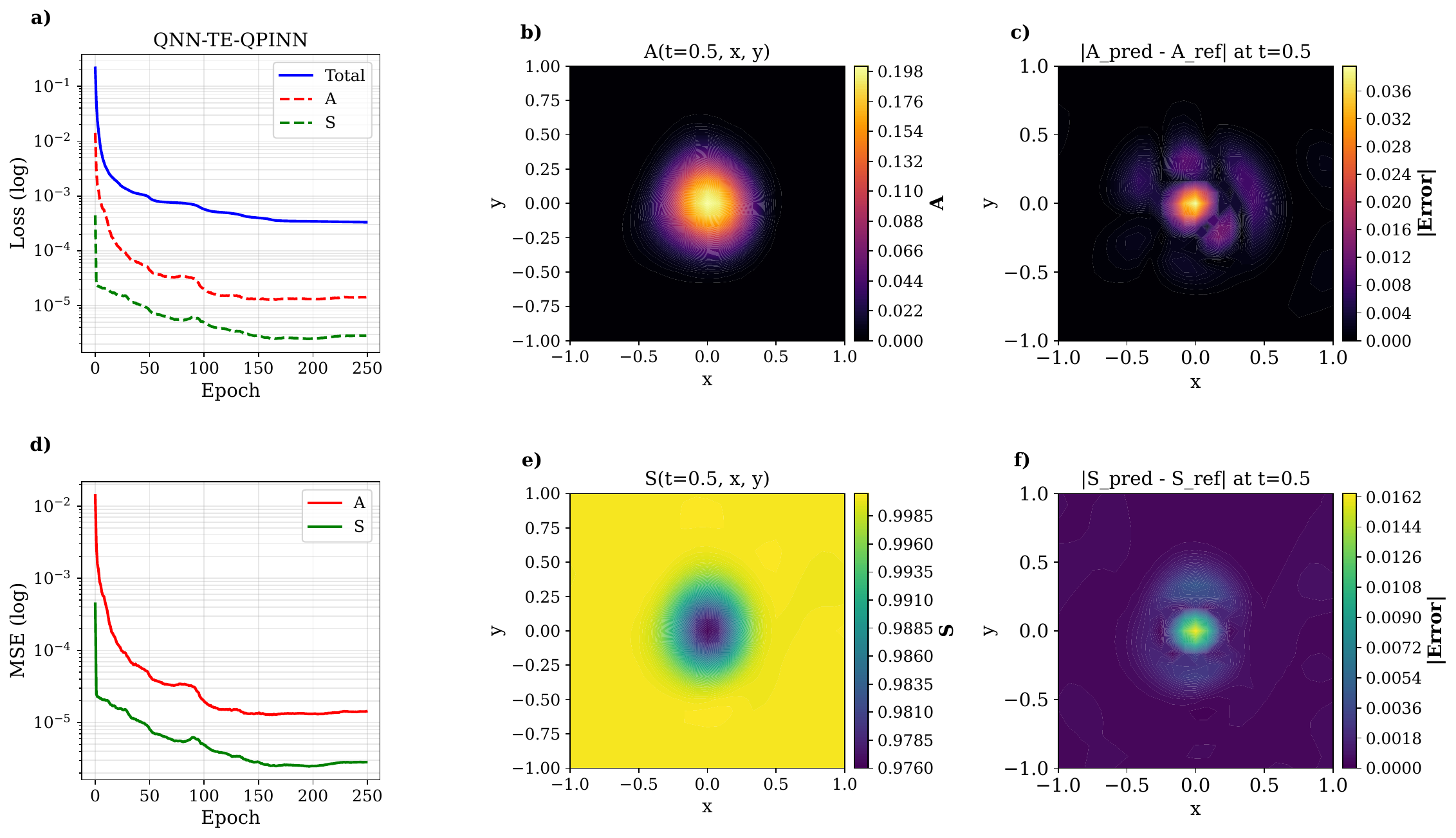}
{Performance of QNN-TE-QPINN trained for 250 epochs with four qubits and ten PQC layers for the two-dimensional RD equation. a) \& d) Total and individual loss, and the MSE for Activator A and Substrate S  of QNN-TE-QPINN. b) \& e) QNN-TE-QPINN model predictions at $t = 0.5$. c) \& f) Absolute error of QNN-TE-QPINN for A and S compared to the RK45 reference solution at $t = 0.5$.
\label{rd_2d_training_fnn_basis}}

To conduct the experiments, we used Python as the primary programming language. We utilized Facebook's classical deep learning library, PyTorch \cite{pytorch}, to design the architecture and classical deep learning neuron layers. For quantum deep learning, we used Xanadu's Pennylane \cite{pennylane} quantum machine learning library to create the trainable quantum circuits. We chose PennyLane because it offers accelerated parameter optimization packages and can be easily integrated with real quantum computing service providers on the cloud, such as IBM, Rigetti, Google and so on, to infer the model after training in a simulator environment. For the experimental hardware, we utilized the REmotely-managed Power Aware Computing Systems and Services (REPACSS) \cite{repacss}, a high-performance computing data center and AI infrastructure prototype sponsored by Texas Tech University, to conduct the research experiments. Specifically, we utilized the GPU compute nodes, each configured with 4 × NVIDIA H100 NVL GPUs with 94 GB of HBM per GPU, 512 GB of memory, and 64 cores per node. The x-TE-QPINNs models were trained using PyTorch's L-BFGS optimizer with default search parameters, applying the strong Wolfe condition. Table \ref{tab:Params} describes the input parameters of the classical neural networks and quantum neural networks that we used in the experiments of this study.

\begin{table}
\centering
\caption{The x-TE-QPINN hyperparameters used in our experiments.}
\label{table}
\setlength{\tabcolsep}{3pt}
\begin{tabular}{p{90pt}  p{30pt} p{90pt}}
\hline
\hline
\textbf{Parameter's Name} & \textbf{Value} & \textbf{Description} \\
\hline
\hline
N\_WIRES &     $2 \rightarrow 8$  &   The number of qubits in the PQC and QNN embedding quantum circuits. \\
N\_LAYERS &    $5 \rightarrow 20$   &   The number of layers in the PQC and QNN embedding quantum circuits.  \\
NEURONS\_FNN &  10   &   The number of neurons in each classical FNN's hidden layer.  \\
HIDDEN\_LAYERS\_FNN &  2   &   The number of hidden layers in the classical FNN.  \\
NEURONS\_PINN &  32   &   The number of neurons in each PINN's hidden layer.  \\
HIDDEN\_LAYERS\_PINN &  4   &   The number of hidden layers in the PINN.  \\
\hline
\hline
\end{tabular}
\label{tab:Params}
\end{table}

\subsection{Data Processing}

Unlike machine learning models trained with traditional data, such as time series in numerical form, texts as vectors, or images as pixels, PINNs and QPINNs utilize data in the form of collocation points. Collocation points are a set of coordinate points $(x_{i},...,x_{j}, t)$—where $i$,..., $j$ denotes the dimensionality of the equation—that are sampled from the computational domain. At these points, the neural network must satisfy the PDE, the boundary conditions, and the initial conditions. In traditional numerical methods, such as the Finite Difference Method (FDM) \cite{Courant1928} or the Finite Volume Method (FVM) \cite{Patankar1980}, the PDE is enforced at all defined grid points. Conversely, in PINNs and QPINNs, the neural networks act as a continuous approximation function, and these collocation points are used to train the network. Fig. \ref{xTE-QPINN-Arch} illustrates the collocation points for the one-dimensional RD equation, $(x_{i=j=0}, t)$. 

\subsection{Inference Solutions}

We employed the fourth- and fifth-order Runge-Kutta method (RK45) \cite{Dormand1980} to solve the RD differential equation, using this solution as the baseline for comparison with the solution predicted by our proposed x-TE-QPINN method. Fig. \ref{ref-sol-1d} shows the RK45 solution as a heatmap on a grid size of $(x, t) = (1, 20)$. RK45 is one of the most advanced and widely used numerical methods for solving ordinary differential equations (ODEs) and PDEs, particularly in complex dynamic systems. It is not only a high-order method but also an adaptive method, which allows it to yield more accurate results than fixed-step methods.  Fig. \ref{ref-sol-2d} illustrates the solution of the two-dimensional RD equation solved using the classical RK45 method, with periodic boundary conditions and random initial conditions.

\subsection{Embedding}

\Figure[t!](topskip=0pt, botskip=0pt, midskip=0pt)[width=2.0\columnwidth]{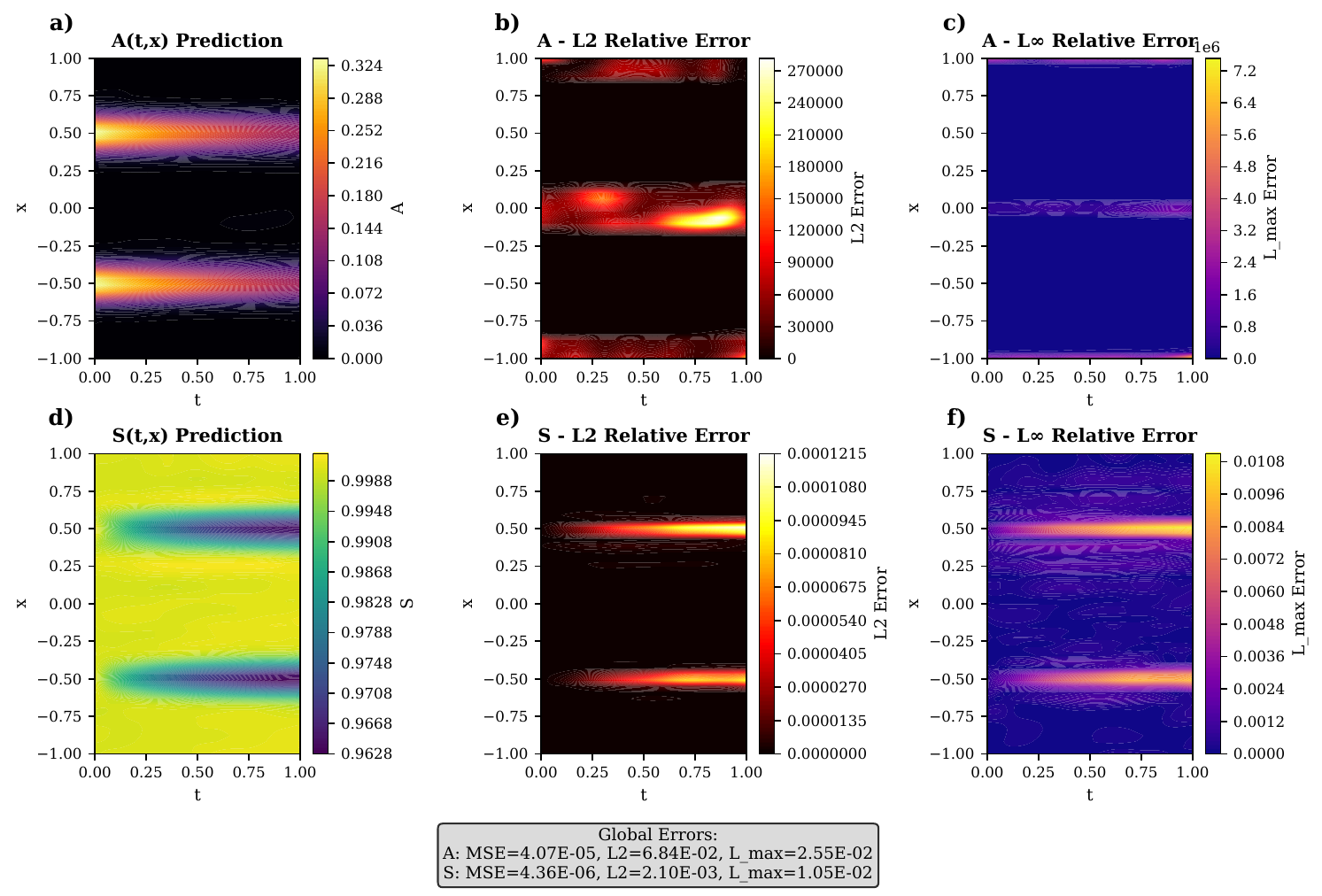}
{Inference results of solving one-dimensional RD equation of the QNN-TE-QPINN model with a configuration of four qubits, ten PQC layers, five embedding layers, after being trained for 300 epochs: a) \& d) Prediction results of activator A and substrate S species. b) \& e) $L_2$ relative error of A and S. c) \& f) $L_\infty$ ($L_{max}$) of A and S.
\label{1d_infer_perf}}

\Figure[t!](topskip=0pt, botskip=0pt, midskip=0pt)[width=2.0\columnwidth]{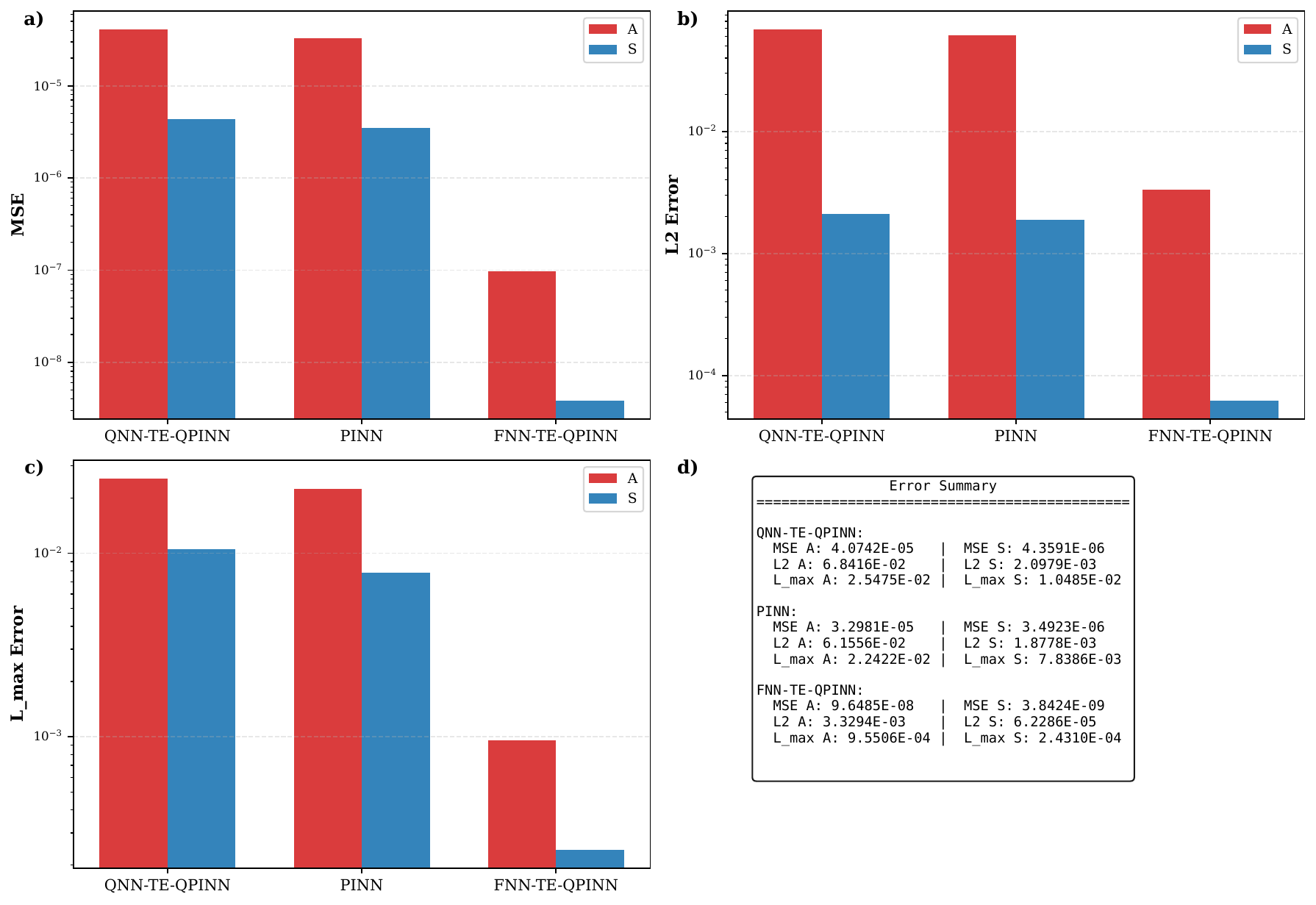}
{This report summarizes the inference results of the PINN, FNN-TE-QPINN, and QNN-TE-QPINN models for solving the one-dimensional equation by focusing on the following evaluations: a) A comparison of performance based on the MSE metric, b) A comparison of results using the $L_2$ relative error, c) A comparison of results based on the $L_\infty$ ($L_{max}$) relative error, and d) A detailed comparison with empirical experimental data.
\label{1d_inf_summary}}

To perform training for x-TE-QPINN, we must embed the classical data into a quantum state, as the quantum kernel requires a quantum state to execute the transformation process, which exploits specific quantum principles and phenomena. This embedding process has been shown to impact the quality of subsequent training \cite{berger2025trainable}. In this architecture, the embedding process will be optimized by a classical or quantum neural network. This neural network learns the parameter $\theta$ alongside the PQC layers, also known as hardware-efficient ansatz (HEA) layers, to predict the most suitable $\theta$ parameter for the embedding, thereby achieving the best predictive performance for the subsequent x-TE-QPINN architecture. Fig. \ref{emb} displays the heatmap result of the embedding function for the input equation for activator A. This embedding function's accuracy is improved by multiplying $\alpha_{i}(\tilde{x}, \tilde{t})$ with $\tilde{t}$ or $\tilde{x}$.

\subsection{Training Performance Analysis}

In this section, we provide a detailed performance analysis of the xTE-QPINN architecture by comparing it against the classical PINN baseline throughout the training process.

\subsubsection{One-Dimensional Reaction-Diffusion Equation}

\Figure[t!](topskip=0pt, botskip=0pt, midskip=0pt)[width=2.0\columnwidth]{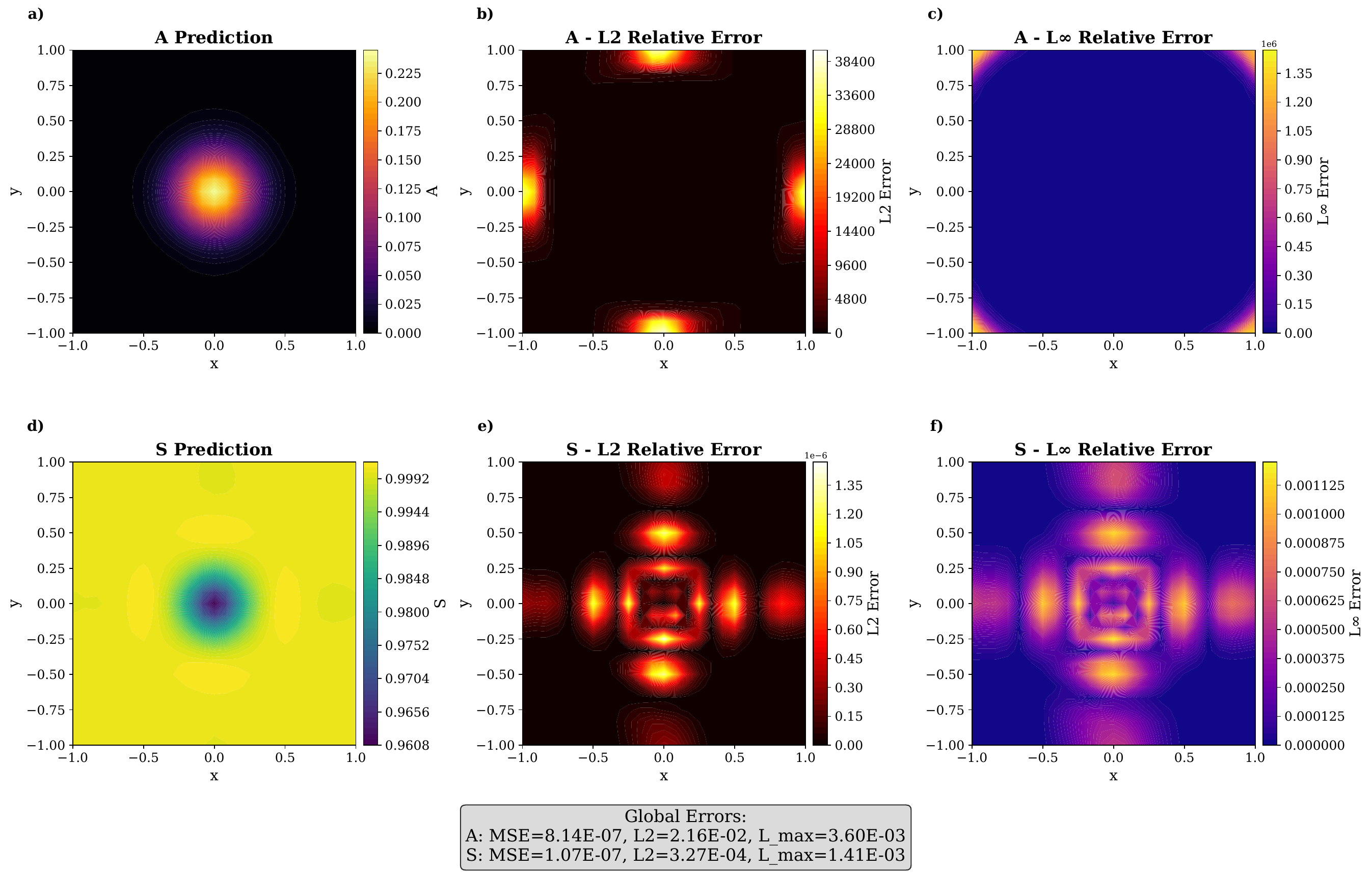}
{Inference results of the two-dimensional RD equation model after training 250 epochs: a) \& d) Prediction of A and S obtained using FNN-TE-QPINN; b) \& e) $L_2$ relative error of A and S, comparing FNN-TE-QPINN predictions with the RK45 reference solution; c) and f) $L_{\infty}$ relative error for A and S, comparing FNN-TE-QPINN predictions with the RK45 reference solution.
\label{rd_2d__fnn_basis_performance}}

A one-dimensional RD system represents the simplest case. We conducted experiments using three models PINN, FNN-TE-QPINN, and QNN-TE-QPINN to solve the PDE of a one-dimensional reaction–diffusion system, evaluating the model’s performance. Parameters used in the simulations were $D_{A} = 1 \times 10^{-5} cm^{2}s^{-1}$ , $D_{S} = 2 \times 10^{-3} cm^{2}s^{-1}$, $k_{1} = 1 m^{-2}s^{-1}$, $k_{2} = 1 s^{-1}$, and $k_{3} = 1 \times 10^{-3} m s^{-1}$. We experimented by gradually increasing the number of qubits and PQC layers in the experiment. We observed that the number of qubits has a significant impact on the GPU's CUDA memory (VRAM) when training QPINN models. Given the current configuration of the REPACSS supercomputer, with a VRAM capacity of 94 GB per GPU, the maximum number of qubits that can be trained for the one-dimensional RD equation is eight. This limitation highlights a significant challenge for classical quantum simulators when attempting to execute quantum algorithm experiments in the current era.

\begin{table}[htbp]
\centering
\caption{Compare the total loss of two models: FNN-TE-QPINN, and QNN-TE-QPINN, when training with a varying number of qubits and a fixed number of epochs at 100 for the one-dimensional RD equation.}
\setlength{\tabcolsep}{3pt}
\begin{tabular}{|p{60pt}|p{70pt}|p{70pt}|}
\hline
\textbf{No of Qubits} & \textbf{FNN-TE-QPINN} & \textbf{QNN-TE-QPINN} \\
\hline
 2 & 7.25E-04  & 1.42E-02 \\
 \hline
 4 & 1.63E-05  & 5.22E-03 \\
 \hline
 6 &  1.41E-05  & 5.09E-03 \\
 \hline
 8 & 5.92E-05  & 2.37E-03 \\
 \hline
\end{tabular}
\label{tab:Perf-1D-qubit}
\end{table}

Table \ref{tab:Perf-1D-qubit} compares the total loss during the training process of the two models, FNN-TE-QPINN and QNN-TE-QPINN, as the number of qubits increases progressively from two to eight. The results show that FNN-TE-QPINN achieves its best quality when trained with six qubits, and the quality decreases as the number of qubits increases to eight. Conversely, QNN-TE-QPINN continues to show a decreasing trend as the number of qubits increases to eight, although the total loss of QNN-TE-QPINN remains higher than that of FNN-TE-QPINN. The higher loss observed in QNN-TE-QPINN is predicted to be due to the fact that both the embedding and learning components are entirely built upon a quantum circuit.

\begin{table}[htbp]
\centering
\caption{Compare the total loss of two models: FNN-TE-QPINN and QNN-TE-QPINN, when training with a varying number of PQC layers and a fixed number of qubits at eight and the number of epochs at 200 for the one-dimensional RD equation.}
\setlength{\tabcolsep}{3pt}
\begin{tabular}{|p{60pt}|p{70pt}|p{70pt}|}
\hline
\textbf{No of PQC layers} & \textbf{FNN-TE-QPINN} & \textbf{QNN-TE-QPINN} \\
\hline
 5  & 6.35E-06 & 2.37E-03 \\
 \hline
 10 & 1.86E-06 & 3.23E-03\\
 \hline
 15 &  9.67E-07  & 2.81E-03 \\
 \hline
 20 &  1.76E-06  & 5.47E-03 \\
 \hline
\end{tabular}
\label{tab:Perf-1D-HEA}
\end{table}

Table \ref{tab:Perf-1D-HEA} illustrates the change in the loss function between the two models, FNN-TE-QPINN and QNN-TE-QPINN, when trained with eight qubits and 200 epochs, while the number of PQC layers is progressively increased from 5 to 20. The results show that both models achieve their best quality at 15 layers and begin to see an increase in the loss function when the number of PQC layers is increased to 20. The experimental results indicate that increasing the number of PQC layers up to an optimal level improves model performance; however, increasing it excessively leads to an increase in circuit depth, which in turn reduces the learning quality.

\begin{table}[htbp]
\centering
\caption{Compare the total loss of three models: PINN, FNN-TE-QPINN, and QNN-TE-QPINN, when training with a varying number of epochs and a fixed number of qubits at eight and the number of PQC layers at five for the one-dimensional RD equation.}
\setlength{\tabcolsep}{3pt}
\begin{tabular}{|p{50pt}|p{30pt}|p{60pt}|p{60pt}|}
\hline
\textbf{No of Epochs} & \textbf{PINN} & \textbf{FNN-TE-QPINN} & \textbf{QNN-TE-QPINN} \\
\hline
 100 & 5.88E-03  & 1.23E-05 & 2.37E-03 \\
 \hline
 200 & 1.55E-03  & 6.35E-06 & 2.37E-03\\
 \hline
 300 & 9.94E-03  & 6.35E-06 & 2.37E-03\\
 \hline
\end{tabular}
\label{tab:Perf-1D-epoch}
\end{table}

Table \ref{tab:Perf-1D-epoch} presents the loss function results for the three models: PINN, FNN-TE-QPINN, and QNN-TE-QPINN, when trained with a progressively increasing number of epochs, ranging from 100 to 300 epochs. The results show that QNN-TE-QPINN achieves convergence when trained with fewer than 100 epochs—specifically, it reaches its optimal loss function value starting from approximately 80 epochs. In contrast, FNN-TE-QPINN requires nearly 200 epochs to achieve this. The PINN model performs slightly worse, requiring close to 300 epochs to converge. When the number of epochs is increased beyond 300, the loss function values for all three models remain constant, indicating that the learning capacity has reached saturation.

To evaluate the quality of the solution produced by PINN, FNN-TE-QPINN, and QNN-TE-QPINN, we further examine the model’s training behavior. Fig. \ref{train_perf_three} illustrates the performance of three models after 300 training epochs with four qubits and 10 PQC layers. The experimental results show that FNN-TE-QPINN yields the best solution quality after training. The line chart in the far-left column represents the total loss for the three models: PINN, FNN-TE-QPINN, and QNN-TE-QPINN. It can be easily observed that FNN-TE-QPINN yields the best training quality, converging after approximately 150 epochs with a total loss reaching 2.86E-06. The progression line of the total loss also exhibits a smooth curve with virtually no sign of oscillation (or 'sawtooth' pattern). This soft line also indicates that the model does not show signs of overfitting. Furthermore, the MSE  for the activator A species decreases to 9.65E-08, which is a minimal value compared to the typical good performance range of other QPINN models, which is around 1E-04. Although it did not achieve as small a total loss as FNN-TE-QPINN, QNN-TE-QPINN has a positive point in that the model achieves convergence very quickly after only about 50 epochs of training. This quick progress demonstrates the superiority of quantum neural network architectures when learning data. Without needing excessive data or prolonged training, quantum neural networks can achieve their maximum generalization capability.

Fig. \ref{TE-QPINN-sol} presents the results obtained after training with 300 epochs of FNN-TE-QPINN model. The solution produced by FNN-TE-QPINN is compared against the reference solution computed using the classical RK45 method. The results show that the FNN-TE-QPINN solution closely matches the classical solution, demonstrating that the model effectively learned the underlying PDE of the one-dimensional RD system. In detail, the far-left column represents the loss function and MSE evolution of the activator A and substrate S species during the training process. The chart shows that the total loss is very smooth and converges between approximately 125 and 150 epochs. The two loss functions for activator A and substrate S show minor sawtooth patterns, but overall, both curves demonstrate a steady and smooth decrease. The two MSE curves for activator A and substrate S exhibit a gradual decrease, reaching saturation at approximately 100 epochs for A and 120 epochs for S. Neither curve shows any increase, confirming that the model does not overfit and achieves good generalization. This result will be explained in more detail in the inference model section after training. The center column contains a heatmap illustrating the prediction results of the FNN-TE-QPINN model during the training process, with the space coordinate $x \in (-1,1)$ and time $t$ ranging from $0$ to $1$. The far-right column shows the absolute error measured between the model's prediction result (taken from the center column) and the reference solution (solved using the RK45 method). The absolute error results indicate that the FNN-TE-QPINN solution matches very well with the solution obtained by the precise method.

\subsubsection{Two-Dimensional Reaction-Diffusion Equation}

In this section, we present the experimental results of the models applied to the two-dimensional RD equation, unlike the one-dimensional case, which involves only one spatial coordinate $x$—where each point is defined by $x$ (such as in a metal rod, pipe, or wire)—each collocation point in two-dimensional space consists of two spatial coordinates $(x, y)$. Here, each point represents a position on a plane $u = u(x, y, t)$, such as animal skin, chemical surfaces, or biological membranes. The Laplacian operator in 2D causes diffusion to occur in two independent directions. This independent two-way diffusion makes the dynamics significantly more complex, resulting in a broader range of patterns generated by the reaction process. Furthermore, the gradients are larger, making the two-dimensional RD equation stiffer and more challenging to solve. In the experimental section for the two-dimensional RD equation, we utilized the same equation parameters as those used in the one-dimensional case.

\begin{table}[htbp]
\centering
\caption{Compare the total loss of three models: PINN, FNN-TE-QPINN, and QNN-TE-QPINN, when training with a varying number of epochs and a fixed number of qubits at six and the number of PQC layers at five for the two-dimensional RD equation.}
\setlength{\tabcolsep}{3pt}
\begin{tabular}{|p{50pt}|p{30pt}|p{60pt}|p{60pt}|}
\hline
\textbf{No of Epochs} & \textbf{PINN} & \textbf{FNN-TE-QPINN} & \textbf{QNN-TE-QPINN} \\
\hline
 100 & 9.86E-04  & 2.19E-05 & 7.66E-04 \\
 \hline
 200 & 5.20E-04  & 1.75E-05 & 7.66E-04\\
 \hline
 300 & 3.36E-04  & 1.99E-05 & 7.66E-04\\
 \hline
 400 & 2.67E-04  & 1.69E-05 & 7.66E-04\\
 \hline
 500 & 2.17E-04  & 1.49E-05 & 7.66E-04\\
 \hline
\end{tabular}
\label{tab:Perf-2D-epoch}
\end{table}

Regarding the number of epochs required for model convergence, similar to the one-dimensional results, QNN-TE-QPINN reaches saturation after nearly 100 epochs of training. In contrast, FNN-TE-QPINN requires up to 300 epochs to achieve its best learning quality. PINN is the model that requires the most extensive training process, as its loss function continues to decrease when the number of epochs is increased from 100 to 500. Detailed results are presented in Table \ref{tab:Perf-2D-epoch}.

\begin{table}[htbp]
\centering
\caption{Compare the total loss of two models: FNN-TE-QPINN, and QNN-TE-QPINN, when training with a varying number of qubits and a fixed number of epochs at 100 for the two-dimensional RD equation.}
\setlength{\tabcolsep}{3pt}
\begin{tabular}{|p{60pt}|p{70pt}|p{70pt}|}
\hline
\textbf{No of Qubits} & \textbf{FNN-TE-QPINN} & \textbf{QNN-TE-QPINN} \\
\hline
 2 & 8.25E-04  & 1.31E-03 \\
 \hline
 4 & 4.02E-05  & 6.13E-04 \\
 \hline
 6 & 2.02E-05  & 7.66E-04 \\
 \hline
\end{tabular}
\label{tab:Perf-2D-qubit}
\end{table}

For the two-dimensional RD equation, the maximum number of qubits that can be simulated on a classical computer with the REPACSS configuration is six qubits. Table \ref{tab:Perf-2D-qubit} presents the comparison of the total loss between the FNN-TE-QPINN and QNN-TE-QPINN versions when experimented with six qubits, five PQC layers, and trained for 100 epochs. The results indicate that FNN-TE-QPINN performs exceptionally well, with the loss function continuing to decrease as the number of qubits increases. In contrast, QNN-TE-QPINN achieves its best results using four qubits, and the model shows signs of degraded learning quality when the number of qubits is increased to six.

\begin{table}[htbp]
\centering
\caption{Compare the total loss of two models: FNN-TE-QPINN and QNN-TE-QPINN, when training with a varying number of PQC layers and a fixed number of qubits at six and the number of epochs at 200 for the two-dimensional RD equation.}
\setlength{\tabcolsep}{3pt}
\begin{tabular}{|p{60pt}|p{70pt}|p{70pt}|}
\hline
\textbf{No of PQC layers} & \textbf{FNN-TE-QPINN} & \textbf{QNN-TE-QPINN} \\
\hline
 5  & 1.75E-05 &  7.66E-04 \\
 \hline
 10 & 6.00E-06 &  2.65E-04\\
 \hline
 15 &  4.08E-06  &  5.82E-04 \\
 \hline
\end{tabular}
\label{tab:Perf-2D-HEA}
\end{table}

Furthermore, the influence of PQC layers on model performance during training varies between FNN-TE-QPINN and QNN-TE-QPINN. The results from Table \ref{tab:Perf-2D-HEA} show that with the current experimental configuration on REPACSS, the maximum number of PQC layers is 15 for the 2D case, compared to 20 for the one-dimensional case. FNN-TE-QPINN demonstrates improved learning capability as the number of PQC layers increases from 5 to 15. Conversely, QNN-TE-QPINN achieves its optimal performance with 10 PQC layers; beyond this point, increasing the number of layers leads to a decline in learning quality. This point can be explained by the fact that the circuit depth of QNN-TE-QPINN increases more rapidly compared to FNN-TE-QPINN, which utilizes a classical neural network component.

Comparing the training results of the three models—PINN, FNN-TE-QPINN, and QNN-TE-QPINN—we observe that FNN-TE-QPINN achieves the superior learning performance (Fig. \ref{rd_2d_model_comparison}). Its total loss, as well as the individual loss functions for A and S, decrease significantly compared to the other two models. Specifically, the total loss for FNN-TE-QPINN reaches 1.31E-05, whereas PINN and QNN-TE-QPINN yield values of 1.56E-04 and 3.32E-04, respectively. Although QNN-TE-QPINN does not yield the lowest loss, its loss curve remains remarkably smooth and exhibits rapid convergence. In contrast, the PINN loss curve is characterized by a 'sawtooth' pattern; notably, the loss for S shows harmonic oscillations during the first 50 epochs. These findings are highly consistent with the results obtained in the one-dimensional case.

Fig. \ref{rd_2d_training_fnn_basis} illustrates the performance of QNN-TE-QPINN trained for 250 epochs, utilizing four qubits for embedding and ten PQC layers for the learning circuit. The results in subfigure (a) show that the total loss function of QNN-TE-QPINN decreases smoothly from over 2.17E-01 at epoch 0 to 3.32E-04 at epoch 249. The individual loss functions for A and S exhibit no "sawtooth" patterns and demonstrate clear signs of convergence from epoch 100 onwards. In subfigure (d), the MSE for A and S remains flat without any upward trend, indicating that the model does not suffer from overfitting. The middle column, comprising subfigures (b) and (e), displays the model's predicted outputs during the validation process. These predictions align closely with the reference solutions obtained via the classical RK45 method, as evidenced by the absolute error shown in subfigures (c) and (f) in the far-right column, measured at time $t = 0.5$.

\subsection{Inferences}
Lastly, the trained model was tested to verify its performance using the PennyLane "default.qubit" simulator. After training the model with suitable epochs and saving the best results, we performed inference on the saved model. For the one-dimensional RD equation,The results for QNN-TE-QPINN are illustrated in Fig. \ref{1d_infer_perf}, where the leftmost column presents the solutions for the RD equations (A and S) as solved by the model. The middle column displays heatmaps of solution accuracy, using the $L_{2}$ relative error to measure distance from the exact solutions obtained via the classical RK45 method. The rightmost column shows the heatmaps for the $L_\infty$ ($L_{max}$) relative error, measuring the maximum relative error produced by the model, along with detailed results measured using three metrics: MSE, $L_{2}$, and $L_{max}$ for A and S. The results indicate that the MSE for both A and S is extremely small, suggesting that, overall, the model has converged very closely to the reference solution. This proximity is clearly visible in the solution heatmaps. An $L_{2}$ relative error of 6.84\% is an acceptable result for a purely quantum circuit used for embedding. The maximum error $L_{max}$ accounts for only 2.5\%, demonstrating significant stability. Furthermore, the narrow gap between the MSE and $L_{max}$ values proves that the model does not suffer from severe local bias or significant localized inaccuracies.

To provide a comprehensive overview of the three models, we aggregate the one-dimensional inference results in Fig. \ref{1d_inf_summary}. We measured three metrics—MSE, $L_2$, and $L_{max}$—for each model and visualized them using bar charts. FNN-TE-QPINN delivers superior performance compared to the other two models, with an MSE for S of 3.8424E-09, an $L_2$ error of 6.2286E-05, and an $L_{max}$ reduced to 2.4310E-04. However, a drawback of this model is the significant discrepancy in errors between A and S across all metrics when compared to the other two architectures.

Fig. \ref{rd_2d__fnn_basis_performance} displays the inference results of the FNN-TE-QPINN model for the two-dimensional RD equation. Subfigures (a) and (d) present the solutions for A and S as predicted by the model. We employed the MSE, $L_2$, and $L_{max}$ relative error to compare these solutions with the RK45 reference solution. In the left, the heatmap results indicate a high degree of alignment, with an MSE of 8.14E-07 for A and 1.07E-07 for S. Subfigure (b) and (e) in the center show the $L_2$ gained 2.16E-02 for A and 3.27E-04 for S.  Furthermore, in the far-right column, we used the $L_{max}$ relative error metric to evaluate the quality of the FNN-TE-QPINN solutions, as illustrated in subfigures (c) and (f). The results show that the $L_{max}$ reaches 3.60E-03 for A and 1.41E-03 for S. These infinitesimal values, ranging from 0.1\% to 0.3\%, demonstrate the exceptionally high accuracy of the FNN-TE-QPINN model, which synergistically combines the power of quantum circuits with classical neural networks for embedding.

\section{Conclusions}

In this study, we have presented the Extended Trainable-Embedding Quantum Physics-Informed Neural Networks method for solving multi-species reaction–diffusion PDEs. We experimentally demonstrated the effectiveness of x-TE-QPINN in solving these dynamics. We also introduced a purely quantum version, QNN-TE-QPINN, which utilizes a quantum neural network for the embedding component, rather than traditional classical neural networks. Our results show that both models perform exceptionally well in solving RD equations, significantly outperforming the conventional classical PINN approach.

The experimental results described in the Results section demonstrate that x-TE-QPINN exhibits superior performance and learning capabilities compared to traditional PINNs when solving the RD equations in both one-dimensional and two-dimensional cases. These quantum models show a distinct advantage in learning physical laws and data patterns compared to their classical counterparts. Specifically, the quantum-enhanced models require significantly fewer training epochs to achieve convergence and yield substantially lower loss function values compared to classical methods. Furthermore, the issue of data embedding—a primary bottleneck hindering the efficiency of hybrid classical-quantum models—is effectively addressed through the integration of neural networks. This solution produces promising results for the future of quantum machine learning models.

Although QNN-TE-QPINN did not perform as well as FNN-TE-QPINN in our experiments, this can be attributed to its fully quantum approach for both embedding and learning. The increased quantum circuit depth is one of the primary factors preventing the loss function from reaching the same level of optimization as classical neural networks. Nevertheless, the significance of this research lies in opening a pathway toward implementing a fully quantum architecture for QPINNs, moving beyond hybrid classical-quantum structures. In the near future, as quantum computing becomes more accessible, designing model architectures that handle both learning circuits and parameter optimization automatically on quantum hardware will be a promising and appropriate direction for future research.

The inherent differences between classical data and quantum states, combined with the measurement process that causes the collapse of the quantum states, often diminish the potential advantages offered by the unique properties of quantum circuits. This mentioned problem represents one of the primary challenges in implementing hybrid classical-quantum models during the noisy intermediate-scale quantum era. This study presents a promising research direction for addressing the challenges at the input-output interface between the classical and quantum domains. Moving forward, we will continue to explore the capabilities of x-TE-QPINN in solving other nonlinear and highly stiff PDEs, which represent a primary strength of these quantum deep learning models. Furthermore, extending the time horizon of the input collocation points leads to the decay problem, also known as gradient pathology \cite{wang2021understanding}, which degrades the accuracy of PINN models—a challenge that x-TE-QPINN also encounters. We intend to investigate this issue in depth to understand its underlying mechanisms and propose optimized solutions to mitigate the decay.

\section*{Appendix}
The source code used in this work can be explored at the accompanying \href{hhttps://codeocean.com/capsule/6966810/}{CodeOcean} capsule.

\bibliographystyle{IEEEtran}
\bibliography{TQE}

\EOD
\end{document}